\documentclass[11pt]{article}
\usepackage[left=1in,top=1in,right=1in,bottom=1in,head=.1in,nofoot]{geometry}

\usepackage{setspace,url,bm,amsmath} 

\usepackage{titlesec} 
\titlelabel{\thetitle.\quad} 

\usepackage[numbers]{natbib}
\usepackage{graphicx}
\usepackage{latexsym}
\usepackage{caption}
\usepackage[margin=20pt]{subcaption}
\usepackage{hyperref}
\usepackage{booktabs}
\usepackage{enumerate} 
\usepackage{mathtools}
\usepackage{amsthm}
\usepackage{amssymb}
\usepackage{amsmath}
\usepackage{color}
\usepackage{booktabs}
\usepackage{comment}
\usepackage{mathrsfs}
\usepackage{natbib}
\usepackage{listings}

\usepackage[table]{xcolor}

\newcommand{\GG}[1]{}

\theoremstyle{definition}
\newtheorem{assume}{Assumption}
\newtheorem*{theorem*}{Theorem}
\newtheorem*{rmk*}{Remark}

\newtheorem{thm}{Theorem}
\newtheorem{prop}{Proposition}
\newtheorem{lem}{Lemma}
\newtheorem{example}{Example}

\newtheorem{cor}{Corollary}

\bibpunct{(}{)}{;}{a}{}{,}  
\usepackage[sc]{mathpazo}

\newcommand{\T}[1]{\textup{#1}}
\newcommand{\bb}[1]{\mathbb{#1}}
\newcommand{\cl}[1]{\mathcal{#1}}
\newcommand{\iid}{\overset{\T{i.i.d.}}{\sim}}
\newcommand{\eqd}{\overset{\T{d}}{=}}
\newcommand{\as}{\overset{\T{a.s.}}{\to}}
\newcommand{\ind}{\overset{\T{d}}{\to}}
\newcommand{\inP}{\overset{\bb{P}}{\to}}
\newcommand{\lmt}[1]{\lim_{#1 \to \infty}}

\newcommand{\op}[1]{#1 #1^\transp}
\newcommand{\ot}[1]{1, \ldots,#1}

\newcommand{\vect}[2]{\begin{pmatrix} #1 \\ #2 \end{pmatrix}}
\newcommand{\vecth}[3]{\begin{pmatrix} #1 \\ #2 \\ #3 \end{pmatrix}}

\DeclareMathOperator{\Var}{Var}

\DeclareMathOperator{\Cov}{Cov}

\DeclareMathOperator{\diag}{diag}
\DeclareMathOperator{\tr}{tr}
\DeclareMathOperator{\Unif}{Unif}
\DeclareMathOperator*{\argmax}{\arg\!\max} 

\DeclareMathOperator*{\plim}{plim}

\def\transp{{\textsc{t}}}

\graphicspath{{Visuals/}}

\allowdisplaybreaks

\begin{document}
\onehalfspacing 
\title{\bf Randomization Tests for Weak Null Hypotheses in Randomized Experiments} 

\author{
Jason Wu and Peng Ding 
\\
Department of Statistics, University of California, Berkeley
\footnote{
Corresponding author: Peng Ding, Email: pengdingpku@berkeley.edu, 425 Evans Hall, Berkeley, CA 94720 USA. We gratefully acknowledge financial support from the U.S. National Science Foundation (DMS RTG \# 1745640 for Jason Wu; DMS grants \# 1713152 and \# 1945136 for Peng Ding).
We thank the Associate Editor, three reviewers, Xinran Li, Guillaume Basse, Joel Middleton, Zach Branson, and Anqi Zhao for helpful comments.
}}

\date{}
\maketitle

\begin{abstract}
The Fisher randomization test (FRT) is appropriate for any test statistic, under a sharp null hypothesis that can recover all missing potential outcomes. However, it is often sought after to test a weak null hypothesis that the treatment does not affect the units on average. To use the FRT for a weak null hypothesis, we must address two issues. First, we need to impute the missing potential outcomes although the weak null hypothesis cannot determine all of them. Second, we need to choose a proper test statistic. For a general weak null hypothesis, we propose an approach to imputing missing potential outcomes under a compatible sharp null hypothesis. Building on this imputation scheme, we advocate a studentized statistic. The resulting FRT has multiple desirable features. First, it is model-free. Second, it is finite-sample exact under the sharp null hypothesis that we use to impute the potential outcomes. Third, it conservatively controls large-sample type I error under the weak null hypothesis of interest. Therefore, our FRT is agnostic to the treatment effect heterogeneity. We establish a unified theory for general factorial experiments and extend it to stratified and clustered experiments.

\end{abstract}
\newpage

\section{Introduction to the Fisher Randomization Test in Experiments}
\subsection{Literature Review}
Randomization is the cornerstone of statistical causal inference \citep[][Section II]{Fisher35}. It creates comparable treatment groups on average. More fundamentally, it justifies the Fisher randomization test (FRT).
Under Fisher's sharp null hypothesis, the treatment does not affect any units whatsoever, and the distribution of any test statistic is known over all randomizations \citep{Fisher35, rubin1980comment, ObsRosenbaum, CausalImbens}.
Therefore, the FRT delivers a finite-sample exact $p$-value. What is more, many parametric and non-parametric tests are approximations to the FRT \citep{eden1933validity, pitman1937significance, kempthorne1952design, box1955permutation, collier1966some, bradley1968distribution, lehmann1975nonparametrics}.

Another formulation of the FRT relies on exchangeability of outcomes under different treatments \citep{pitman1937significance, hoeffding1952large, Romano90}. They called this formulation a ``permutation test''.
\citet{kempthorne1969behaviour} accentuated the importance of the treatment assignment mechanism to justify the FRT, without assuming that the outcomes are exchangeable. \citet{rubin1980comment} extended the FRT using \citet{Neyman23}'s potential outcomes.
He defined a null hypothesis to be sharp if it can determine all missing potential outcomes. One of his insights was that any test statistic has a known distribution under a sharp null hypothesis, and therefore the FRT is finite-sample exact.

Randomized experiments are increasingly popular in the social sciences \citep{duflo2007using, gerber2012field, CausalImbens, ATHEY201773}. In such applications, testing sharp null hypotheses may not answer the researchers' queries.
They often want to test weak null hypotheses that the treatment has zero effects on average. The ideal testing procedure must leave room for treatment effect heterogeneity.
Unfortunately, weak null hypotheses cannot determine all missing potential outcomes, even though the distributions of test statistics depend on them in general. Consequently, simple FRTs may not be directly applicable for testing weak null hypotheses.

Having the FRT test weak null hypotheses is a delicate task. Although sometimes we can still wield the same FRTs, we need to modify the interpretations when the null hypothesis is not sharp \citep{rosenbaum1999reduced, rosenbaum2001effects, Rosenbaum:2003en, caughey2017beyond}.
Not all FRTs can preserve type I errors for weak null hypotheses even asymptotically. The famous Neyman--Fisher controversy ties into this issue for randomized block designs and Latin square designs \citep{neyman1935statistical, sabbaghi2014comments}.
\citet{gail1996design} and \citet{lin2017placement} gave empirical evidence from simulations, and \citet{DD18} gave a theoretical analysis of the one-way layout. Two strategies exist for using FRTs to test weak null hypotheses.
The first strategy realizes that weak null hypotheses become sharp given appropriate nuisance parameters. It maximizes the $p$-values over all values of the nuisance parameters or their confidence sets \citep{Nolen:2011fk, Rigdon:2015, li2016exact, DingTEV16}.
However, it can be computationally expensive and lacks power when the nuisance parameters are high dimensional. The second strategy uses conditional FRTs.
It relies on partitioning the space of all randomizations, and in some subspaces, certain test statistics have known distributions under the weak null hypotheses \citep{athey2015exact, basse2017exact}. It can be restrictive and is not applicable in general settings.

\subsection{Our Contributions}
We propose a strategy for testing a general hypothesis in a completely randomized factorial experiment. The null hypothesis asserts that certain average factorial effects are zero. It is therefore weak and cannot determine all missing potential outcomes. Our strategy has two components.

First, we specify a sharp null hypothesis. It must imply the weak null hypothesis being tested and be compatible with the observed data. Treatment-unit additivity holds under this sharp null hypothesis.
In particular, it implies constant factorial effects of and beyond the weak null hypothesis. Under this sharp null hypothesis, we can impute all missing potential outcomes.

Second, we use the FRT with a studentized test statistic. Like other test statistics, its sampling distribution depends on unknown potential outcomes in general. Thus, this distribution is outside our grasp.
Fortunately, the FRT generates a proxy distribution under the above sharp null hypothesis. This proxy distribution stochastically dominates the unknown one asymptotically.
The stochastic dominance relationship between them enables us to construct an asymptotically conservative test. Therefore, for testing the weak null hypothesis, we recommend the FRT with the studentized statistic.
Barring studentization, the FRT may not control type I error even asymptotically. We examine several existing test statistics that exhibit this unwanted behavior.

The idea of studentization already surfaces in the literature. \citet{neuhaus1993conditional}, \citet{Janssen97}, \citet{janssen1999testing}, \citet{janssen2003bootstrap} and \citet{Romano13} conducted permutation tests with studentization. These tests assumed that the outcomes are independent draws.
In our formulation, the random treatment assignment drives the statistical inference on fixed potential outcomes. We do not assume any exchangeability of outcomes.

In this particular setting, our theory transmits many new features. First, the sampling distribution of the studentized statistic is not asymptotically pivotal, unlike in an independent samples setting. Rather, the approximate distribution generated by the FRT is.
Second, the FRT is conservative for the weak null hypothesis. This aspect of finite-population causal inference \citep{Neyman23, CausalImbens} was absent in the literature on permutation tests.
Third, studentizaion helps us achieve better first order accuracy, i.e., to control asymptotic type I error. \citet{babu1983inference} and \citet{hall1988theoretical}, on the other hand, used it for better second order accuracy in the bootstrap.

The bootstrap is another resampling method for testing weak null hypotheses. Relative to the bootstrap, FRTs have an additional advantage of being finite-sample exact under sharp null hypotheses.
Although the bootstrap has been a workhorse for many other statistical problems, \citet{imbens2018causal} recently fused its ideas with finite population causal inference.

\subsection{Organization and Notation}
Let us preview how the rest of the paper is organized. Section \ref{s:fw} lays out the potential outcomes framework, FRTs, and the null hypotheses of interest. 
Section \ref{s:TestStat} formalizes what kind of test statistic can be used with the FRT to test weak null hypotheses. It then gives our advocated test statistic that meets the criterion, and also other popular test statistics that do not.
Section \ref{s:sc} gives various examples of special cases covered by our results in Section \ref{s:TestStat}. Section \ref{s:ext} shows how our results, with some modifications, can be extended to other classes of experiments.
Section \ref{s:sim} uses simulations to look at the finite sample behavior of the FRT with studentization, complementing our asymptotic theory. Section \ref{s:appl} demonstrates further the application of our results by using data from real-world randomized experiments.
Section \ref{s:disc} wraps up our paper. The online supplementary material has the proofs for all of our results.

Let $1_n $ and $0_n$ be vectors of $n$ $1$'s and $0$'s, respectively. Let $1(\cdot)$ denote the indicator that an event happens.
Let $A \succeq 0$ and $A \succ 0$ if $A$ is positive semi-definite and positive definite, respectively. Write $A \succeq B$ if $A-B \succeq 0$. For a diagonalizable matrix $A$, let $\lambda_j(A)$ be its $j$-th largest eigenvalue. Let diag$\{\cdot \}$ be a diagonal or block-diagonal matrix.
If $(X_N)$ is a sequence of random variables indexed by $N$, write $X_N \ind X$, $X_N \inP X$, $X_n \as X$ for convergence in distribution, probability, and almost surely (often abbreviated ``a.s.''), respectively. For convergence in probability, we may also write $\plim_{N \to \infty}X_N=X$.  For random vectors or matrices, the same notation denotes such convergence, entry by entry.
Let $\Pi_N$ denote the set of permutations of $\{ \ot{N}\}$. Let $\pi$ denote a generic element of $\Pi_N$, which is a mapping from $\{ \ot{N}\}$ to itself. Let Unif$(\Pi_N)$ denote the uniform distribution over $\Pi_N$.
Random variable $B$ stochastically dominates $A$, written $A \leq_\T{st}B$, if their cumulative distribution functions $F_A(x)$ and $F_B(x)$ satisfy $F_A(x) \geq F_B(x)$ for all $x$. Let $\xi_1,\xi_2,\ldots$ be independent and identically distributed (i.i.d.) $\cl{N}(0,1)$ random variables.

\section{Framework}\label{s:fw}
\subsection{Completely Randomized Experiments}\label{sec::cre}
We adhere to the potential outcomes framework \citep{Neyman23, Rubin74}. Let $Y_i(j)$ be the response of unit $i$ if it receives treatment $j$, where $i=\ot{N}$ and $j=\ot{J}$. Vectorize $Y_i =(Y_i(1), \ldots,Y_i(J))^\transp$.
The means of the potential outcomes are $\bar{Y}(j)= \sum_{i=1}^NY_i(j)/N$, vectorized as $\bar{Y}=(\bar{Y}(1), \ldots,\bar{Y}(J))^\transp$.
The covariance between the potential outcomes is $S(j,k)= \sum_{i=1}^N \{ Y_i(j)- \bar{Y}(j)\} \{ Y_i(k)- \bar{Y}(k) \} /(N-1)$, which is a variance if $j=k$. The covariance matrix $S$ has the $(j,k)$-th entry $S(j,k)$.

Let $W_i \in \{ \ot{J}\}$ represent the treatment that unit $i$ actually receives, and define the indicator $W_i(j)=1(W_i=j)$. The $W = (W_1, \ldots,W_N)$ are generated according to a completely randomized experiment (CRE).
The experimenter picks $N_1, \ldots,N_J \geq 2$ that sum to $N$, and assigns treatments randomly so that any realization satisfies $\sum_{i=1}^NW_i(j)=N_j$ for $j=\ot{J}$, and has probability $\prod_{j=1}^JN_j!/N!$.

Unit $i$'s observed outcome is $Y_i^\T{obs}=Y_i(W_i)= \sum_{j=1}^JW_i(j)Y_i(j)$. So the observed means are $\hat{\bar{Y}}(j)= \sum_{i=1}^NW_i(j)Y_i^\T{obs}/N_j$, vectorized as $\hat{\bar{Y}}=( \hat{\bar{Y}}(1), \ldots, \hat{\bar{Y}}(J))^\transp$.
The observed variances are $\hat{S}(j,j)= \sum_{i=1}^NW_i(j) \{ Y_i^\T{obs}- \hat{\bar{Y}}(j) \}^2/(N_j-1)$, which is the sample analog of $S(j,j)$. Because $Y_i(j)$ and $Y_i(k)$ are not jointly observable, there is no sample analog for $S(j,k)$.
In general, we cannot estimate $S(j,k)$ consistently for $j \neq k$. For regularity, we assume $S(j,j)>0$ and $\hat{S}(j,j)>0$ for all $W=(W_1, \ldots,W_N)^\transp$. 

\subsection{Fisher Randomization Tests}
The \emph{Fisher Randomization Test} (FRT) was formulated by \citet{Fisher35} to analyze experimental data. Several flavors of it exist \citep{pitman1937significance, hoeffding1952large, basu1980randomization, Romano90}. We adopt that of \citet{rubin1980comment}. It arises from the potential outcomes described in Section \ref{sec::cre}.

\citet{Rubin05} called the potential outcome matrix $\{  Y_i(j):i=\ot{N},j=\ot{J}\}$ the Science Table. He termed a null hypothesis \emph{sharp} if it, along with the observed data, can determine all the missing items in the Science Table. 
A test statistic is a function of the observed data and the null hypothesis. Under a sharp null hypothesis, any test statistic has a known distribution. In particular, we can cycle through the possible values of $W$, and for each obtain the corresponding realization of observed data, and then compute the value of the test statistic.
In this manner, the test statistic's distribution becomes accessible, as does a $p$-value. FRTs are therefore finite-sample exact for testing sharp null hypotheses, no matter the test statistic or data generating process for the potential outcomes \citep{ObsRosenbaum, CausalImbens}.
In essence, randomization is fundamental for statistical inference. It justifies the FRT, and guarantees the validity of the resulting $p$-value.

Practitioners typically brand sharp null hypotheses as too restrictive. In a general factorial experiment, our mission is to test
\begin{equation}\label{e:HNx}
H_{0 \T{N}}(C,x):C \bar{Y}=x, 
\end{equation}
where $x \in \bb{R}^m$ and $C \in \bb{R}^{m \times J}$ is a full row rank contrast matrix, i.e., $C1_J=0_m$. We pay extra attention to hypotheses where $x=0_m$, but study general $x$ for completeness.
A \emph{weak} hypothesis is any that is not sharp by the definition of \citet{Rubin05}. The hypothesis \eqref{e:HNx} is therefore weak. It is also referred to as an average/Neyman null hypothesis.
It only confines the averages of the potential outcomes. Meanwhile, a sharp/strong/Fisher null hypothesis confines all individual potential outcomes.

Notwithstanding that the FRT is designed for sharp null hypotheses, we ask whether it can test \eqref{e:HNx} also. The FRT mandates that all potential outcomes be filled out. We do so aided by an artificial sharp null hypothesis. A sensible one is
\begin{equation}\label{e:HFx}
H_{0 \T{F}}(C,x, \tilde{C},\tilde{x}): \vect{C}{\tilde{C}}Y_i= \vect{x}{\tilde{x}}\T{ for }i=\ot{N},
\end{equation} 
where the matrix $(C^\transp, \tilde{C}^\transp,1_J)$ is invertible. When $m=J-1$, $\tilde{C}$ and $\tilde{x}$ are empty, as $(C^\transp,1_J)$ already form an invertible square matrix.
When $m<J-1$, we can construct $\tilde{C}$ from $C$ and $1_J$ by Gram--Schmidt orthogonalization. We are then to select $\tilde{x}\in \bb{R}^{J-m-1}$. Whatever we select here does not matter asymptotically, as we see later.
For null hypotheses \eqref{e:HNx} where $x=0_m$, we can go with $\tilde{x}=0_{J-m-1}$ to get the classical sharp null hypothesis of no individual effects whatsoever. Intuitively, the piece $CY_i=x$ of \eqref{e:HFx} is ``of'' the weak null hypothesis \eqref{e:HNx}, and the piece $\tilde{C}Y_i= \tilde{x}$ is ``beyond'' it.
The hypothesis \eqref{e:HFx} induces two key features. The first is the weak null hypothesis \eqref{e:HNx}. The second is strict additivity, i.e., $Y_i(j)-Y_i(k)$ does not depend on the unit $i$, for $j,k=\ot{J}$.

With the sharp null hypothesis \eqref{e:HFx} and some test statistic $T$ that ideally can capture possible deviation from \eqref{e:HFx}, the FRT proceeds as follows. 
\begin{enumerate}[FRT-1.]
\item Calculate $T$ from $\{ W_i,Y_i^\T{obs}:i=\ot{N}\}$.

\item\label{step::impute} Impute potential outcomes:  
\[ Y_i^*= \vecth{Y_i^*(1)}{\vdots}{Y_i^*(J)}=z+(Y_i^\T{obs}-z_{W_i})1_J,
\T{ where }z= \vecth{z_1}{\vdots}{z_J}= \vecth{C}{\tilde{C}}{1_J^\transp}^{-1}\vecth{x}{\tilde{x}}{0}, \]
or, equivalently, $Y_i^*(j)= Y_i^\T{obs}+z_j-z_{W_i}\T{ for }j=\ot{J}$.

\item\label{step::permute} For a permutation $\pi \in \Pi_N$, compute $Y_{\pi,i}^\T{obs} = \sum_{j=1}^JW_{\pi(i)}(j)Y_i^*(j)$ and calculate $T_\pi$ from $\{ W_{\pi(i)},Y_{\pi,i}^\T{obs}:i=\ot{N}\}$ the same way $T$ was calculated. 

\item\label{step::pv} The $p$-value is
$(N!)^{-1}\sum_{\pi \in \Pi_N}1(T_\pi \geq T)$. 

\end{enumerate}

As a sanity check, the imputed potential outcomes in FRT-\ref{step::impute} satisfy \eqref{e:HFx} and $Y_i^*(W_i)=Y_i^\T{obs}$ for all $i$.
Given the Science Table, every realization of treatment assignment $W$ produces data $\{ W_i,Y_i^\T{obs}:i=\ot{N}\}$. Henceforth, we call the values of $T$ that can possibly emerge from these data the \emph{sampling distribution} of $T$.
Conditioning on the original data $\{ W_i,Y_i^\T{obs}:i=\ot{N}\}$, we can fill out missing potential outcomes with FRT-\ref{step::impute}. We call the set of values $\{ T_\pi:\pi \in \Pi_N \}$ defined in FRT-\ref{step::permute} the \emph{randomization distribution} of $T$.
Since this distribution depends on the original data, whose randomness comes solely from $W$, we denote this distribution by $T_\pi|W$.

Under the sharp null hypothesis that the treatment truly does not affect any unit whatsoever, the FRT just described reduces to the classical permutation test. In particular, in FRT-\ref{step::impute}, all potential outcomes are equal to the observed outcome $Y_i^\T{obs}$, and in FRT-\ref{step::permute}, we just need to permute the treatment assignment $W$ because $Y_{\pi,i}^\T{obs} = Y_i^\T{obs}$ for every unit $i$. Under this sharp null hypothesis, the FRT and permutation test are numerically identical. There is an isomorphism between the two in this sense. In general, the FRT admits a broader class of null hypotheses and experimental designs than the permutation test.

Step FRT-\ref{step::pv} conveys that the FRT $p$-value is a right-tail probability. A larger value of $T$ embodies a larger deviation from the null hypothesis.  
Even if $N!$ is too large for a manageable exact computation of the $p$-value, we are able to fall back on random i.i.d. draws from $\Pi_N$ to approximate the $p$-value in FRT-\ref{step::pv} subject to Monte Carlo error. We are thus always at liberty to sample randomly from the randomization distribution.

For any test statistic $T$, the $p$-value in FRT-\ref{step::pv} is valid under \eqref{e:HFx}. Our overarching goal is to investigate whether the FRT can still control type I error for testing $H_{0 \T{N}}(C,x)$.
Roughly speaking, this turns out to be affirmative asymptotically granted an appropriate test statistic $T$. Before continuing, let us be specific that the FRT with $T$ conservatively controls type I error at level $\alpha$ if $\bb{P}\left\{  (N!)^{-1}\sum_{\pi \in \Pi_N}1(T_\pi \geq T) \leq \alpha \right\} \leq \alpha$. 
In words, the true probability of a conservative test incorrectly rejecting the null hypothesis is never greater than the nominal significance level. For conciseness, when we say a test controls type I error, we do not always mention explicitly that it does so conservatively.

\subsection{Asymptotics for Finite Population Inference}
We have contended that the exact sampling distribution of $T$ depends on unknown potential outcomes under $H_{0 \T{N}}(C,x)$ in general. Instead of finite sample theory, we embrace an asymptotic theory. This gives us a feasible approximation to the sampling distribution of $T$.
Imagine a sequence of finite populations of potential outcomes. For each $N \geq 2J$, we fix in advance $N_1, \ldots,N_J \geq 2$. Independently across $N$, we generate $W$ according to a CRE, from which we get $Y_i^\T{obs}$ and calculate a test statistic. We denote a sequence indexed by $N$ with $N \to \infty$ by $(\cdot)$ or $(\cdot)_{N \geq 2J}$.
Technically, we should index finite population quantities by $N$, and also index observed quantities by $N_1, \ldots,N_J$. For cleaner notation, and with a nod to the precedent of earlier authors, we drop these extra subscripts, unless to emphasize the dependence on $N$. We now state our assumptions on the sequence of potential outcomes.
\begin{assume}\label{AsuA}
The sequence $(N_j/N)$ converges to $p_j \in (0,1)$ for all $j= \ot{J}$.
The sequences $( \bar{Y}_N)$ and $(S_N)$ converge to $\bar{Y}_\infty < \infty$ and $S_\infty$, where $S_\infty$ has finite entries and positive main diagonal entries. Further, $\lmt{N}\max_{j=\ot{J}}\max_{i=\ot{N}}\big\{ Y_i(j)- \bar{Y}(j) \big\}^2/N=0$.
\end{assume}
\begin{assume}\label{AsuB}
Same as Assumption \ref{AsuA} except the last sentence is replaced by: Further, there exists an $L< \infty$ such that $\sum_{i=1}^N \{ Y_i(j)- \bar{Y}(j) \}^4/N \leq L$ for all $j=\ot{J}$ and $N \geq 2J$.
\end{assume}
\begin{prop}\label{p:BimpA}
Assumption \ref{AsuB} implies Assumption \ref{AsuA}.
\end{prop}
The design of experiments often guarantees the existence of $p_j\in (0,1)$ because all treatment groups have comparable sizes in realistic scenarios. We can weaken the existence of $\bar{Y}_\infty$ and $S_\infty$ by standardizing the potential outcomes. Just as we drop $N$, we might drop subscripts $\infty$. 
For instance, $S$ can mean either the finite population covariance matrix or its limiting value, which will be clear from context. Intuitively, Assumption \ref{AsuA} requires more than two moments, and Assumption \ref{AsuB} requires four moments. Assumption \ref{AsuB} is thus stronger than Assumption \ref{AsuA}.
Below are our principal asymptotic tools, which are consequences of \cite{DingCLT}.
\begin{prop}\label{p:inP}
Under Assumption \ref{AsuA},
$\hat{\bar{Y}}- \bar{Y}=(\hat{\bar{Y}}(1)- \bar{Y}(1), \ldots, \hat{\bar{Y}}(J)- \bar{Y}(J))^\transp \inP 0_J$, and $\hat{S}(j,j)- S(j,j)  \inP 0 \T{ for }j=\ot{J}$.
\end{prop} 

\begin{prop}\label{t:Ding5}
Under Assumption \ref{AsuA}, $N^{1/2}( \hat{\bar{Y}}- \bar{Y}) \ind \cl{N}(0_J,V)$, where
\begin{align}\label{e:ObsCov}
V= \lmt{N}N \cdot \Cov( \hat{\bar{Y}})= \lmt{N}\begin{pmatrix}
\frac{N-N_1}{N_1}S(1,1) & -S(1,2) & \cdots & -S(1,J) \\
-S(2,1) & \frac{N-N_2}{N_2}S(2,2) & \cdots & -S(2,J) \\
\vdots & \vdots & \ddots & \vdots \\
-S(J,1) & -S(J,2) & \cdots& \frac{N-N_J}{N_J}S(J,J) \end{pmatrix}.
\end{align}
\end{prop}
The limiting distribution in Proposition \ref{t:Ding5} depends on unknown quantities. We need to estimate $N \cdot \Cov( \hat{\bar{Y}})$. This covariance, however, depends on $S(j,k)$ ($j \neq k$), which do not have unbiased estimators in general. Prompted by \cite{Neyman23}, we estimate the main diagonal by:
\[ \hat{D}=N \cdot \diag\left\{ \hat{S}(1,1)/N_1, \cdots, \hat{S}(J,J)/N_J \right\} \succ 0. \]
Proposition \ref{p:inP} implies
\begin{equation}\label{e:Dob}
\hat{D}\inP D= \diag\left\{ S(1,1)/p_1, \cdots,S(J,J)/p_J \right\} \succ 0.
\end{equation}
Because $V=D-S \preceq D$, the estimator $\hat{D}$ is asymptotically conservative for $N \cdot \Cov( \hat{\bar{Y}})$ in the sense that $\lmt{N}N \cdot \Cov( \hat{\bar{Y}}) \preceq \plim_{N \to \infty}\hat{D}$. We will encounter this notion time after time. \cite{AronowVar14} brings up tight bounds for covariance estimation in treatment-control randomized experiments with $J=2$. Their results suggest that we can further improve the estimator $\hat{D}$. Nevertheless, we will show that $\hat{D}$ suffices for our goal of testing \eqref{e:HNx} with FRTs. 

\section{Test Statistics}\label{s:TestStat}
We return to our main endeavor: whether the FRT with a test statistic $T$ can control type I error when testing $H_{0 \T{N}}(C,x)$. The next proposition demarcates precisely what kind of $T$ can accomplish this goal.
\begin{prop}\label{p:suitable}
Consider testing $H_{0 \T{N}}(C,x)$. The FRT with test statistic $T$ controls type I error at any level if under $H_{0 \T{N}}(C,x)$, the sampling distribution of $T$ is stochastically dominated by its randomization distribution, that is, $T \leq_\T{st}T_\pi |W$.
\end{prop}
To test $H_{0 \T{N}}(C,x)$, we use a test statistic $T$, but look upon its randomization distribution $T_\pi |W$ as the reference null distribution. The $p$-value in FRT-\ref{step::pv} is the probability that $T_\pi |W$ is at least the observed value of $T$.
If $T \leq_\T{st}T_\pi |W$, then any quantile of the asymptotic distribution of $T_\pi | W$ is at least that of $T$. Consequently, we have conservative tests at any level.

It is quite burdensome to ensure a meaningful test statistic satisfies the criterion of Proposition \ref{p:suitable}. 
For a candidate statistic $T$, we instead settle for ascertaining whether its randomization distribution stochastically dominates its sampling distribution asymptotically under $H_{0 \T{N}}(C,x)$ for almost all sequences of $W$. Henceforth, we call $T$ \emph{proper} if so.

\subsection{Studentized statistic}
We advocate using the following studentized statistic in the FRT:
\begin{equation}\label{e:X2}
X^2=N(C \hat{\bar{Y}}-x)^\transp (C \hat{D}C^\transp )^{-1}(C \hat{\bar{Y}}-x).
\end{equation}
It is a Wald-type statistic that has a conservative covariance estimator $C \hat{D}C^\transp$ for $N^{1/2}(C \hat{\bar{Y}}-x)$.

Studentized statistics have appeared alongside permutation tests when the outcomes are independent samples. \cite{Romano90} was aware of the problem of test statistics that were not studentized in two-sample tests.  
For \cite{Janssen97}, studentization was an avenue in the Behrens--Fisher problem to control the type I error. \cite{Romano13} studied the same phenomenon when the parameter being cared about could be more general than the mean.
\cite{Pauly15} and \cite{Kon15} embraced an equivalent studentized statistic in general factorial experiments with independent samples. In the aforementioned settings, studentization works because the test statistic is asymptotically pivotal.

As for us, $X^2$ is itself not asymptotically pivotal. Rather, it is stochastically dominated by a pivotal distribution. This is a key reason it is exactly the statistic we seek based on Proposition \ref{p:suitable}. We now formally state our main result that $X^2$ is proper.
\begin{thm}\label{t:CYx}
If Assumption \ref{AsuA} holds, then under $H_{0 \T{N}}(C,x)$, $X^2 \ind \sum_{j=1}^ma_j \xi_j^2$, where each $a_j \in [0,1]$. If Assumption \ref{AsuB} holds, and $\pi \sim \Unif( \Pi_N)$, then $X_\pi^2|W \ind \chi_m^2$ a.s.
\end{thm}
Immediate from this theorem is that the FRT using $X^2$ controls the asymptotic type I error under $H_{0 \T{N}}(C,x)$. This test also retains finite sample exactness under the sharp null hypothesis \eqref{e:HFx}. As a result, it is robust for inference on two classes of null hypotheses.

Asymptotically, under $H_{0 \T{N}}(C,x)$, neither the sampling nor randomization distribution of $X^2$ depends on $\tilde{C}$ or $\tilde{x}$, so the choice of $\tilde{x}$ does not matter.
The randomization distribution also does not depend on $H_{0 \T{N}}(C,x)$. A violation of $H_{0 \T{N}}(C,x)$ is likely to inflate the value of $X^2$ but not the values of $X_\pi^2|W$. An appealing consequence of this fact is that the FRT using $X^2$ has power.

Echoing \cite{Romano13} and \cite{Pauly15}, one purpose of studentization for us is to control type I error. Yet, for us, the FRT using $X^2$ is asymptotically conservative, while the corresponding test in an independent samples setting is asymptotically exact.
This stems from our potential outcomes framework: $\{ \hat{\bar{Y}}(1), \ldots, \hat{\bar{Y}}(J) \}$ do not have vanishing correlations, even asymptotically.

Theorem \ref{t:CYx} inspires another asymptotically conservative test besides the FRT. We can reject $H_{0 \T{N}}(C,x)$ if the observed value of $X^2$ exceeds the $1- \alpha$ quantile of $\chi_m^2$.
We call this alternative to the FRT the $\chi^2$ approximation. This is computationally efficient without Monte Carlo. The FRT has an additional property.
It is concurrently finite-sample exact for the sharp null hypothesis \eqref{e:HFx}. Our simulations and practical data examples compare these two classes of tests empirically.

\subsection{Box-Type Statistic}
We now steer toward an alternative statistic, one found in \cite{Brunner97}:
\begin{equation}\label{e:BoxStat}
B=N \hat{\bar{Y}}^\transp M \hat{\bar{Y}}/ \tr(M \hat{D}),
\end{equation}
where $M=C^\transp (CC^\transp )^{-1}C$ is the projection matrix onto the row space of $C$. Because we will deem it as not proper in our context, we can restrict the discussion to $x=0_m$.

Under independent sampling, \cite{Brunner97} approximated the asymptotic behavior of $B$ by an $F$ distribution through ideas from \cite{Box54}, and called it a Box-type statistic. Their simulations found it to enjoy superior empirical small sample properties under their framework.

For our problem, the next result states the behavior of $B$. Recall $V$ in \eqref{e:ObsCov} and define $P= \diag(p_1,\ldots,p_J)$.
\begin{thm}\label{t:Brunner}
If Assumption \ref{AsuA} holds, then under $H_{0 \T{N}}(C,0_m)$, $B \ind \sum_{j=1}^m \lambda_j(MV) \xi_j^2/ \tr(MD)$. If Assumption \ref{AsuB} holds and $\pi \sim \Unif( \Pi_N)$, then $B_\pi |W \ind \sum_{j=1}^m \lambda_j(MP^{-1}) \xi_j^2/ \tr(MP^{-1})$ a.s.
\end{thm}
The asymptotic mean of $B$ is $\sum_{j=1}^m \lambda_j(MV)/ \tr(MD) \leq 1$ because $V \preceq D$, and the asymptotic mean of $B_\pi |W$ is $\sum_{j=1}^m \lambda_j(MP^{-1})/ \tr(MP^{-1})=1$.
Therefore, the former mean does not exceed the latter. This is necessary but not sufficient for the stochastic dominance criterion of Proposition \ref{p:suitable}, which does not hold.
Hence, the FRT with the Box-type statistic cannot control type I error in general, even asymptotically. This is the subject of a later simulation.

There are two situations where $B$ is proper: equal variances, and testing a one-dimensional hypothesis.
\begin{cor}\label{c:Brunner}
Under Assumption \ref{AsuB}, if $S(1,1)= \cdots=S(J,J)$, then $B$ meets the criterion of Proposition \ref{p:suitable} asymptotically. If $C$ is a row vector, then $B=X^2$.
\end{cor}

\subsection{Statistics from Ordinary Least Squares}
Ordinary least squares (OLS) tools are widespread in the analysis of experimental data \citep[e.g.,][]{morris2010design}. We insert $J$-treatment randomized experiments into the realm of linear models.
We do this by encoding the treatments with dummy variables in the design matrix $\cl{X}= \diag(1_{N_1}, \ldots,1_{N_J})$. The response vector consists of the corresponding observed outcomes from treatment groups $\ot{J}$.  
The OLS coefficients are the entries of $\hat{\bar{Y}}$, which has estimated covariance matrix $\hat{\sigma}^2( \cl{X}^\transp \cl{X})^{-1}$, where $\hat{\sigma}^2=(N-J)^{-1}\sum_{i=1}^N\sum_{j=1}^JW_i(j) \{ Y_i^\T{obs}- \hat{\bar{Y}}(j) \}^2$ is the mean residual sum of squares. The classical $F$ statistic for testing \eqref{e:HNx} is then
\begin{align}\label{e:Fstat}
F=(C \hat{\bar{Y}})^\transp \{ \hat{\sigma}^2C( \cl{X}^\transp \cl{X})^{-1}C^\transp \}^{-1}C \hat{\bar{Y}}/m.
\end{align}
We do not stipulate the usual assumptions of linear regression, but just want a test statistic for the FRT.

We first record a peculiar situation where $F$ is identical to the Box-type statistic $B$. This result will be valuable for our simulations and practical data examples.
\begin{prop}\label{p:BeqF}
$B=F$ if $N_1= \cdots=N_J$ and $M=C^\transp (CC^\transp )^{-1}C$ has the same entries along its main diagonal.
\end{prop}
Except for the scaling by $m$ and the presence of $\hat{\sigma}^2$ in place of each $\hat{S}(j,j)$, $F$ is identical to $X^2$. This pooled variance estimate $\hat{\sigma}^2$ is problematic for the $F$ statistic, spurring it to fall short of the criterion of Proposition \ref{p:suitable}, as we formalize next.
\begin{thm}\label{t:LSF}
If Assumption \ref{AsuA} holds, then under $H_{0 \T{N}}(C,0_m)$, $m \cdot F \ind \sum_{j=1}^m \lambda_j \big( CVC^\transp ( \bar{S}CP^{-1}C^\transp )^{-1}\big) \xi_j^2$ where $\bar{S}= \sum_{j=1}^Jp_jS(j,j)$. If Assumption \ref{AsuB} holds and $\pi \sim \Unif( \Pi_N)$, then $m \cdot F_\pi |W \ind \chi_m^2$ a.s.
\end{thm}
The classical linear model assumes a constant treatment effect for all units \citep{kempthorne1952design}. This necessitates equal variances under all treatment levels. Yet, such homoscedasticity is not built into the potential outcomes framework. The assumptions underlying the $F$ statistic are not compatible with the potential outcomes framework in general. If the potential outcomes do have equal variance, then it is not surprising that $F$ is proper.
\begin{cor}\label{c:LSF}
Under Assumption \ref{AsuB}, if $S(1,1)= \cdots=S(J,J)$, then $F$ meets the criterion of Proposition \ref{p:suitable} asymptotically.
\end{cor}
Huber--White covariance estimation for the OLS coefficients is frequently quoted as a fix to the classical $F$ statistic.
Econometricians are especially inclined to such an estimate of the covariance when the linear model is possibly misspecified or the error terms are heteroscedastic.
Define the residual $\hat{\epsilon}_i=Y_i^\T{obs}- \hat{\bar{Y}}(W_i)$. The Huber--White estimator for $N \cdot \Cov( \hat{\bar{Y}})$ is
\[ \begin{split}
\hat{D}_\T{HW}=& N( \cl{X}^\transp \cl{X})^{-1}\cl{X}^\transp \diag\left\{ \hat{\epsilon}_1^2, \ldots, \hat{\epsilon}_N^2 \right\} \cl{X}( \cl{X}^\transp \cl{X})^{-1}\\
=& N \cdot \diag\left\{ \frac{N_1-1}{N_1^2}\hat{S}(1,1), \ldots, \frac{N_J-1}{N_J^2}\hat{S}(J,J) \right\} .
\end{split}\]
If we replace $\hat{\sigma}^2( \cl{X}^\transp \cl{X})^{-1}$ by $\hat{D}_\T{HW}$ in \eqref{e:Fstat} and dismiss the scaling by $m$, we get
\[ X_\T{HW}^2=N(C \hat{\bar{Y}})^\transp (C \hat{D}_\T{HW}C^\transp)^{-1}C \hat{\bar{Y}}. \]
$\hat{D}_\T{HW} $ is nearly identical to $\hat{D}$ if $N_j \approx N_j-1$ for $j=\ot{J}$. Therefore, $X_\T{HW}^2$ is asymptotically akin to $X^2$. By this, the Huber--White covariance estimator successfully repairs the $F$ statistic.

\section{Special Cases}\label{s:sc}
Section \ref{s:TestStat} devises a strategy for testing weak null hypotheses in general experiments. The contents there speak directly to many worthwhile settings.

\subsection{One-Way Analysis of Variance with Multi-Valued Treatments}\label{s:anova}
In the one-way analysis of variance (ANOVA), the goal is to test $H_{0 \T{N}}: \bar{Y}(1)= \cdots= \bar{Y}(J)$. It is a special case of the null hypothesis \eqref{e:HNx} with $x=0_{J-1}$ and any contrast matrix $C \in \bb{R}^{(J-1) \times J}$ for instance $C=(1_{J-1},-I_{J-1})$. Here, $m=J-1$, which spares us from having to construct $\tilde{C}$ or select $\tilde{x}$.

We impute potential outcomes in FRT-\ref{step::impute} as $Y_i^*(j)=Y_i^\T{obs}$ for $i=\ot{N}$ and $j=\ot{J}$ under $H_{0 \T{F}}: Y_i(1)= \cdots= Y_i(J)$, for $i=\ot{N}$. To test $H_{0 \T{F}}$, \cite{Fisher25} crafted the statistic
\begin{equation}\label{e:DD18-F}
F=  \frac{\sum_{j=1}^JN_j \{ \hat{\bar{Y}}(j)- \bar{Y}_\cdot^\T{obs} \}^2/(J-1)}{\sum_{j=1}^J(N_j-1) \hat{S}(j,j)/(N-J)}, 
\T{ where }\bar{Y}_\cdot^\T{obs}= \frac{1}{N}\sum_{i=1}^NY_i^\T{obs}.
\end{equation}
He argued that $F_{J-1,N-J}$ approximates the sampling distribution of $F$. \cite{DD18} attested that \eqref{e:DD18-F} is not proper but
\begin{equation}\label{e:DD18-X2}
X^2= \sum_{j=1}^J \frac{N_j}{\hat{S}(j,j)}\{ \hat{\bar{Y}}(j)- \bar{Y}_S^\T{obs}\}^2,
\T{ where }\bar{Y}_S^\T{obs}= \frac{\sum_{j=1}^JN_j \hat{\bar{Y}}(j)/ \hat{S}(j,j)}{\sum_{j=1}^JN_j/ \hat{S}(j,j)}
\end{equation}
is for testing $H_{0 \T{N}}$ with the FRT. See \citet{schochet2018multi} for a related discussion.

It is immediate from the next proposition that our framework encompasses these results as special cases.
\begin{prop}\label{p:eqX2}
In the one-way ANOVA, the $X^2$ in \eqref{e:X2} and \eqref{e:DD18-X2} coincide, as do the $F$ in \eqref{e:Fstat} and \eqref{e:DD18-F}.
\end{prop}

\subsection{Treatment-Control Experiments}\label{s:tc}
In the treatment-control setting, $J=2$, and unit $i$ either receives the treatment (then $Y_i^\T{obs}=Y_i(1)$) or control (then $Y_i^\T{obs}=Y_i(2)$).
A parameter we might inquire about is the average treatment effect $\tau= \bar{Y}(1)- \bar{Y}(2)$. The weak null hypothesis is $H_{0 \T{N}}(C,0): \tau=0$. This matches \eqref{e:HNx}, where $C=(1,-1)$ is a row vector. Thus, treatment-control is a special case of the one-way layout of Section \ref{s:anova}.
A popular statistic is $| \hat{\tau}|$, where $\hat{\tau}= \hat{\bar{Y}}(1)- \hat{\bar{Y}}(2)$ is the sample difference-in-means of outcomes. However, \citet{DD18} showed that $| \hat{\tau}|$ is not proper for testing $H_{0 \T{N}}$.
\begin{cor}\label{c:tc}
In the treatment-control setting,
\begin{equation}\label{e:X2-tc}
X^2=B= \frac{\{ \hat{\bar{Y}}(1)- \hat{\bar{Y}}(2) \}^2}{\hat{S}(1,1)/N_1+ \hat{S}(2,2)/N_2}= \frac{\hat{\tau}^2}{\hat{S}(1,1)/N_1+ \hat{S}(2,2)/N_2} = t^2,
\end{equation}
where $t$ is the studentized statistic, i.e., \cite{Neyman23}'s estimator of the average causal effect divided by its standard error.
Under Assumption \ref{AsuB}, for almost all sequences of $W$, $B=X^2$ can asymptotically control type I error, but $F$ and $| \hat{\tau}|$ cannot, unless $N_1=N_2$ or $S(1,1)=S(2,2)$.
\end{cor}
Because $t$ is a monotone transform of $X^2$, the FRT with $|t|$ is asymptotically conservative in the finite population setup. It also leads to exact type I errors for the sharp null hypothesis $H_{0\text{F}}: Y_i(1) = Y_i(2)$ for all $i.$
Not only is the statistic $| \hat{\tau}|$ not proper, but it also has other ``paradoxical'' shortcomings \citep{DingParadox}; see also the comment of \citet{loh2017apparent}.
Corollary \ref{c:tc} declares that a balanced design can salvage the $F$ and $| \hat{\tau}|$ statistics, even without homoscedasticity. Perhaps counter to intuition, this protection does not endure when $J>2$, as our simulations will soon demonstrate.

\subsection{Trend Tests}\label{s:trend}
Our perspective has been on type I error under null hypotheses without specifying alternative hypotheses. In experiments for dose-response relationships, we have ordered treatment $1\leq \cdots \leq J$ and often specify the null and alternative hypotheses as $H_{0 \T{N}}$ and $H_{1 \T{N}}: \bar{Y}(1) \leq  \cdots \leq \bar{Y}(J)$ with at least one strict inequality.
We can still carry forward the results in Section \ref{s:anova} on ANOVA. Power might shrink for the test if we do not account for the ordering of the dose-response relationship. Motivated by \citet{armitage1955tests} and \citet{page1963ordered}, we first choose doses $(a_1,\ldots, a_J)$ for treatment levels $( \ot{J})$. Then the test statistic
$ 
r=  \sum_{j=1}^Ja_j \{ \hat{\bar{Y}}(j)- \bar{Y}_\cdot^\T{obs}\} =C \hat{\bar{Y}}
$
is plausible, where $C=(a_1-a_+N_1/N, \ldots,a_J-a_+N_J/N) \in \bb{R}^{1 \times J}$ is a contrast vector, and $a_+= \sum_{j=1}^Ja_j$. In effect, we are testing $H_{0\T{N}}(C,0):C \bar{Y}=0$. Previous theory suggests that $r$ is not proper but the studentized statistic is:
\[ 
t= \frac{ C \hat{\bar{Y}}}{ (C\hat{D}C^\transp /N)^{1/2} }
=  \frac{ \sum_{j=1}^Ja_j \{  \hat{\bar{Y}}(j)- \bar{Y}_\cdot^\T{obs}\}    }{  \{   \sum_{j=1}^J    (a_j - a_+ N_j/N)^2 \hat{S}(j,j)/N_j  \}^{1/2} }.
\]
Note that under $H_{0 \T{N}}$, we impute all missing potential outcomes as $Y_i^\T{obs}$ for each unit $i$, albeit we fix a particular contrast vector $C$ to construct the studentized statistic.
Moreover, in this case, we conduct a one-sided test, rejecting $H_{0 \T{N}}$ if $t$ is larger than the $1- \alpha$ quantile of its randomization distribution.

\subsection{Binary Outcomes}
The theory for $X^2$ statistics does not insist that the outcome be of a particular type as long as the regularity conditions hold.
In particular, it applies directly to binary outcomes. However, binary outcomes have a special feature that $S(j,j) = N \bar{Y}(j)\{ 1- \bar{Y}(j) \} /(N-1)$, i.e., the mean $\bar{Y}(j)$ determines the variance $S(j,j)$.
Therefore, under the null hypothesis $H_{0 \T{N}}: \bar{Y}(1)= \cdots = \bar{Y}(J)$, the variances are all the same too: $S(1,1)= \cdots =S(J,J)$.
For binary outcomes, the difference-in-means statistic $| \hat{\tau}|$ for $J=2$ in Section \ref{s:tc}, the $F$ statistic for general $J$ in Section \ref{s:anova}, and the $r$ statistic in Section \ref{s:trend} are all proper for testing $H_{0 \T{N}}$.
As pointed out by \citet{DingParadox}, for this weak null hypothesis, we do not need studentization to guarantee correct asymptotic type I error.
However, this does not hold for general weak null hypotheses $H_{0 \T{N}}(C,x)$ of binary potential outcomes because $C \bar{Y}=x$ does not imply they have equal variances.  
In general, we always recommend using $X^2$.

\subsection{$2^K$ Factorial Designs}
$2^K$ factorial designs seek to analyze $K$ binary treatment factors simultaneously. In total, we have $J = 2^K$ possible treatment combinations. \cite{DasFact15} tied these designs and the potential outcomes framework together. We summarize this setup.
To do so, it is helpful to introduce the model matrix $G \in \{ \pm 1 \}^{(J-1) \times J}$. Let $*$ denote the component-wise product. \cite{FactLu16b} constructed the rows of $G$, which we call $g_1^\transp, \ldots,g_{J-1}^\transp$, as follows: 
\begin{itemize}
\item for $j=\ot{K}$, let $g_j^\transp$ be $-1_{2^{K-j}}^\transp,1_{2^{K-j}}^\transp$ repeated $2^{j-1}$ times;
\item the next $\binom{K}{2}$ values of $g_j$'s are $g_{k(1)}*g_{k(2)}$ where $k(1) \neq k(2) \in \{ \ot{K}\}$;
\item the next $\binom{K}{3}$ are component-wise products of triplets of distinct $g_1, \ldots,g_K$, etc;
\item the bottom row is $g_{J-1}=g_1* \cdots *g_K$.
\end{itemize}
The matrix $G$ has rows orthogonal to each other and to $1_J$, i.e., $GG^\transp =J \cdot I_{J-1}$ and $G1_J=0_{J-1}$. Let $\tilde{G}\in \{ \pm 1 \}^{K \times J}$ be the first $K$ rows of $G$. Call its columns $z_1, \ldots,z_J$, which are the possible treatment combinations. An example elucidates the setup.
\begin{example}\label{ex:fact2}
When $K=2$, we have
\[ G= \begin{pmatrix}
-1 & -1 & 1 & 1 \\
-1 & 1 & -1 & 1 \\
1 & -1 & -1 & 1
\end{pmatrix}= \vecth{g_1^\transp}{g_2^\transp}{g_3^\transp}= \vect{\tilde{G}}{g_3^\transp}= \begin{pmatrix}
z_1 & z_2 & z_3 & z_4 \\
1 & -1 &  -1 & 1
\end{pmatrix}. \]
The four possible treatment combinations are $z_1=(-1,-1)^\transp$, $z_2=(-1,1)^\transp$, $z_3=(1,-1)^\transp$, and $z_4=(1,1)^\transp$. We read these off from the first two rows of $G$.\qed
\end{example}
The rows of $G$ define factorial effects. Namely, $g_1, \ldots,g_K$ correspond to main effects, $g_{K+1}, \ldots,g_{K+ \binom{K}{2}}$ correspond to two-way interactions, etc, and $g_{J-1}$ corresponds to the $K$-way interaction.
Let $Y_i(j)=Y_i(z_j)$ be the response of unit $i$ if it receives the treatment combination $z_j$. Then we can transfer our previous notation to $2^K$ factorial designs.
The general factorial effect for unit $i$ indexed by $g_j$ is $\tau_{ij}=2g_j^\transp Y_i/J$, and the corresponding average factorial effect is $\tau_j= \sum_{i=1}^N \tau_{ij}/N=2g_j^\transp \bar{Y}/J$.
Vectorize these quantities: $\tau_i =( \tau_{i1}, \ldots, \tau_{i,J-1})^\transp =2GY_i/J$ and $\tau =( \tau_1, \ldots, \tau_{J-1})^\transp =2G \bar{Y}/J$.

We may perform inference on $\tau$ or any subset of its entries. Let $A= \{a(1), \ldots,a(m) \} \subseteq \{ \ot{J-1}\}$ be the target subset, and let $C \in \{ \pm 1 \}^{m \times J}$ have rows $g_{a(1)}^\transp, \ldots,g_{a(m)}^\transp$. 
Then $\tau_A =( \tau_{a(1)}, \ldots, \tau_{a(m)})^\transp =2C \bar{Y}/J$. Testing whether $\tau_A=2x/J$ is equivalent to testing $H_{0 \T{N}}(C,x)$. The FRT with $X^2$ is proper.
The factorial design stimulates a natural choice of $\tilde{C}$ for the imputation step FRT-\ref{step::impute}. We let $g_j^\transp$ be a row of $\tilde{C}$ whenever $j \notin A$.

\citet{FactLu16b} discussed both randomization-based and regression-based inferences for $2^K$ factorial designs. He fixated on point estimation and proposed using the Huber--White covariance estimator.
We have likewise highlighted that it is imperative to use the Huber--White covariance estimator and the $F$ statistic together in the FRT.

\subsection{Hodges--Lehmann Estimation}
Up to this stage, our developments have been on hypothesis testing. Drawing upon the duality between testing and estimation, our previous results shed light on the estimation of $C \bar{Y}$.
This strategy is sometimes referred to as Hodges--Lehmann estimation \citep{HL63, ObsRosenbaum}. For a fixed $x$, we can by means of the FRT obtain a $p$-value for the null hypothesis $H_{0 \T{N}}(C,x)$. Let us denote this $p$-value by $p(x)$ to delineate its dependence on $x$.

The Hodges--Lehmann point estimator $\hat{\tau}_\T{HL}$ for $C \bar{Y}$ is the $x \in \bb{R}^m$ that results in the least significant $p$-value for testing $H_{0 \T{N}}(C,x)$. In symbols, $\hat{\tau}_\T{HL}\in \argmax_{x \in \bb{R}^m}p(x)$.
Note that $x=C \hat{\bar{Y}}$ implies $X^2=0$, which in turn implies $p(x)=1$. Thus $\hat{\tau}_\T{HL}=C \hat{\bar{Y}}$, the usual unbiased estimator. Because $X^2$ is proper, the duality between hypothesis testing and confidence sets assures the following corollary.

\begin{cor}\label{c:CI-FRT}
For $\alpha \in (0,1)$ and almost all sequences of $W$, an asymptotically conservative $(1- \alpha)$ confidence set for $C \bar{Y}$ is
$
\T{CR}_\alpha = \big\{ x \in \bb{R}^m:p(x)> \alpha \big\}, 
$
in the sense that $\lmt{N}\bb{P}\{ C \bar{Y}\in \T{CR}_\alpha \} \geq 1-\alpha$.
\end{cor}

Determining $\T{CR}_\alpha$ can be computationally intensive, so it is expedient to have the asymptotic approximation 
\begin{equation}\label{e:CIapprox}
\T{CR}_\alpha \approx \left\{ x:N(C \hat{\bar{Y}}-x)^\transp (C\hat{D}C^\transp )^{-1}(C \hat{\bar{Y}}-x) \leq \chi_{m,\alpha}^2 \right\},
\end{equation}
where $\chi_{m,\alpha}^2$ is the $1- \alpha$ quantile of $\chi_m^2$. Because the $X^2$ statistic is a quadratic form, $\T{CR}_\alpha$ is an ellipsoid centered at $C \hat{\bar{Y}}$.
The set $\T{CR}_\alpha$ can serve either directly as a $1- \alpha$ approximate confidence set or as an initial guess in searching for the exact confidence region by inverting FRTs. We undertake this later by a simulation.

\subsection{Testing Inequalities} 

FRTs can also handle hypotheses of inequalities:
\begin{equation}\label{e:HNx-ineq}
\tilde{H}_{0 \T{N}}(C,x): C \bar{Y}\geq x.
\end{equation}
We commence at the case where $C \in \bb{R}^{1 \times J}$ is a row vector with $C1_J=0$, and $x \in \bb{R}$ is a scalar.
\begin{example}
In the two-sample problem with $J=2$, we can test $\bar{Y}(2)- \bar{Y}(1) \geq 0$: whether treatment level $1$ results in smaller outcomes than treatment level $2$ on average. In this case, $C = (-1,1)$ and $x=0$.\qed
\end{example}
\begin{example}\label{Ex:3arm}
In a gold standard design for three arms, let level $1$ be the placebo control, level $2$ be the active control, and level $3$ be the experimental treatment. Suppose that smaller outcomes are more desirable, and we know that $\bar{Y}(2)> \bar{Y}(1)$ from previous studies.
Given $\Delta >0$, the goal is to test the hypothesis $\bar{Y}(1)- \bar{Y}(3) \leq \Delta \{ \bar{Y}(1)- \bar{Y}(2) \}$. When $\Delta >1$, this is a superiority test, and when $\Delta \in (0,1)$, this is a non-inferiority test \citep{Mutze3arm}.
This null hypothesis is equivalent to $\tilde{H}_{0 \T{N}}(C,0): ( \Delta -1) \bar{Y}(1)- \Delta \bar{Y}(2)+ \bar{Y}(3) \geq 0$ with $C=( \Delta -1,- \Delta, 1)$.\qed
\end{example}
To impute the missing potential outcomes, we pretend that the null hypothesis is $H_{0 \T{N}}(C,x)$ and utilize \eqref{step::impute} as we did before.
The statistic $X^2$ is not suitable here because it is intended for two-sided tests. For instance, $X^2$ can be large, even under $\tilde{H}_{0 \T{N}}(C,x)$. Instead we use a truncated statistic $t_+= \max(t,0)$ where
\[ t=N^{1/2}(x-C \hat{\bar{Y}})/(C \hat{D}C^\transp )^{1/2}. \]
The FRT with $t$ also works for $p$-values at most $0.5$. \citet{Mutze3arm} used the special case of $t$ in the setting of Example \ref{Ex:3arm}. We choose $t_+$ so that Proposition \ref{p:suitable} directly covers our situation. We summarize the results below.
\begin{cor}\label{c:CYx-ineq}
Consider testing $\tilde{H}_{0 \T{N}}(C,x)$ in \eqref{e:HNx-ineq}, where $C \in \bb{R}^{1 \times J}$ and $x \in \bb{R}$. If Assumption \ref{AsuA} holds, then under $H_{0 \T{N}}(C,x)$ in \eqref{e:HNx}, we have $t \ind \cl{N}(0,a)$ for some $a \in [0,1]$.
If Assumption \ref{AsuB} holds and $\pi \sim \Unif( \Pi_N)$, then $t_\pi |W \ind \cl{N}(0,1)$ a.s. In particular, the FRT with test statistic $t_+$ can asymptotically control type I error under $\tilde{H}_{0 \T{N}}(C,x)$ a.s.
\end{cor}
When $C \in \bb{R}^{m \times J}$ and $x \in \bb{R}^m$ for $m>1$, we can interpret \eqref{e:HNx-ineq} as component-wise inequalities. Neither $X^2$ nor $t_+$ are acceptable when $m>1$. An elementary workaround is to test each component using $t_+$ and apply a Bonferroni correction.

\subsection{Cluster-Randomized Experiments}
In many applied settings, the $N$ units are partitioned into $L$ clusters (e.g., classrooms in educational studies, villages in public health studies). All units belonging to a cluster must receive the same treatment. A cluster-randomized experiment assigns treatments to clusters, i.e. it is a CRE treating clusters as units.
For $l=\ot{L}$, let $\breve{W}_l \in \{ \ot{J}\}$ represent the treatment that cluster $l$ receives, and define the indicator $\breve{W}_l(j)=1( \breve{W}_l=j)$. There are $L!/ \prod_{j=1}^JL_j!$ possible realizations of $( \breve{W}_1, \ldots, \breve{W}_L )$. 
The mechanism of treatment assignment to clusters is identical to that to individuals in a CRE.

\cite{MiddletonCl15} stressed that we cannot implement the same analysis as if we had a CRE on the $N$ units. For instance, $\hat{\bar{Y}}(j)$ is no longer an unbiased estimator for $\bar{Y}(j)$ if the cluster sizes vary. 
Both \cite{MiddletonCl15} and \cite{DingCLT} advised a CRE-like analysis. Let $X_i \in \{ \ot{L}\}$ represent the cluster membership of unit $i$. Define cluster level aggregated potential outcomes $\{ A_l(j):l=\ot{L},j=\ot{J}\}$, where $A_l(j)= \sum_{i=1}^N1(X_i=l)Y_i(j)$.
Define $A_l=(A_l(1), \ldots,A_l(J))^\transp$, $A_l^\T{obs}$, $\bar{A}=( \bar{A}(1), \ldots, \bar{A}(J))^\transp$, $\hat{\bar{A}}=( \hat{\bar{A}}(1), \ldots, \hat{\bar{A}}(J))^\transp$ to align with our previous notation for a CRE.
Aggregated potential outcomes resolve the problem of unbiased estimation of $\bar{Y}$: $\bb{E} (L \hat{\bar{A}}/N) = L \bar{A}/N = \bar{Y}$. Define $\hat{S}_A(j,j)= \sum_{l=1}^L \breve{W}_l(j) \{ A_l^\T{obs}- \hat{\bar{A}}(j) \}^2/(L_j-1)$ and $\hat{D}_A=L \cdot \diag\{ \hat{S}_A(1,1)/L_1, \ldots, \hat{S}_A(J,J)/L_J \}$. We revise the $X^2$ statistic as
\[ X_A^2=L(C \hat{\bar{A}}-Nx/L)^\transp (C \hat{D}_AC^\transp )^{-1}(C \hat{\bar{A}}-Nx/L). \]
Then Theorem \ref{t:CYx} tells us that $X_A^2$ is proper for $H_{0 \T{N}}(C,x)$ as $L \to \infty$ if Assumption \ref{AsuB} holds for the aggregated potential outcomes.

\section{Extensions}\label{s:ext}
\subsection{Stratified Randomized Experiments}
We extend previous results to the stratified randomized experiment (SRE), also called the randomized block design. The overall setup from the CRE still applies, but now for each unit we also observe an associated covariate $X_i \in \{ \ot{H}\}$. Thus, our data are $\{ Y_i^\T{obs},X_i,W_i:i= \ot{N}\}$. 
The treatment does not affect this covariate. The $W_i$'s remain the sole source of randomness. For $h=\ot{H}$, the $h$-th stratum consists of all units $i$ where $X_i=h$, whose size is $N_{[h]}= \sum_{i=1}^N1(X_i = h)$ and proportion is $\omega_{[h]}=N_{[h]}/N$.
For $h=\ot{H}$ and $j=\ot{J}$, the experimenter predetermines the sample sizes $N_{[h]j}= \sum_{i=1}^N1(X_i=h,W_i=j) \geq 2$. In a SRE, we assign treatments within each stratum just as we did in a CRE, and independently among different strata \citep{CausalImbens}.

To define within-stratum means and covariances, we mirror previous notation. For $h=\ot{H}$, the mean vector is $\bar{Y}_{[h]}\in \bb{R}^J$, which has $j$-th entry $\bar{Y}_{[h]}(j)= \sum_{i=1}^N1(X_i=h)Y_i(j)/N_{[h]}$.
The covariance $S_{[h]}$ has $(j,k)$-th entry $S_{[h]}(j,k)= \sum_{i=1}^N1(X_i= h) \{ Y_i(j)- \bar{Y}_{[h]}(j) \} \{ Y_i(k)- \bar{Y}_{[h]}(k) \} /(N_{[h]}-1)$. We impose Assumption \ref{AsuB} on all strata.
\begin{assume}\label{Asu-bl}
For $h=\ot{H}$, (1) $\lmt{N}N_{[h]}/N= \omega_{[h]}\geq 0$ and $\lmt{N}N_{[h]j}/N_{[h]}=p_{[h]j}>0$; (2) the sequences $( \bar{Y}_{[h]})$ and $(S_{[h]})$ converge to $\bar{Y}_{[h] \infty}$ and $S_{[h] \infty}$;
(3) the matrix $S_{[h] \infty}$ has strictly positive main diagonal entries; (4) there exists an $L< \infty$ such that $\sum_{i=1}^N1(X_i=h) \{ Y_i(j)- \bar{Y}_{[h]}(j) \}^4/N_{[h]}\leq L$ for all $N$ and $j=\ot{J}$.
\end{assume}
We do not distinguish between Assumptions \ref{AsuA} and \ref{AsuB} in the SRE for convenience. Tolerating a tiny abuse of notation, $\omega_{[h]}$ stands for both $N_{[h]}/N$ and its limit.
The sample mean vector is $\hat{\bar{Y}}_{[h]}\in \bb{R}^J$, which has $j$-th entry $\hat{\bar{Y}}_{[h]}(j)= \sum_{i=1}^N1(X_i=h,W_i=j)Y_i^\T{obs}/N_{[h]j}$.
The sample variance is $\hat{S}_{[h]}(j,j)= \sum_{i=1}^N1(X_i=h,W_i=j) \{ Y_i^\T{obs}- \hat{\bar{Y}}_{[h]}(j) \}^2/(N_{[h]j}-1)$.
Under Assumption \ref{Asu-bl}, we have from Proposition \ref{t:Ding5} that, inside stratum $h$, the standardized stratum-wise sample mean $N_{[h]}^{1/2}( \hat{\bar{Y}}_{[h]}-  \bar{Y}_{[h]})$ is asymptotically Normal with mean $0$ and a covariance we denote $V_{[h]}$. A conservative estimator for $V_{[h]}$ is 
\[ \hat{D}_{[h]}=N_{[h]}\cdot \diag\{ \hat{S}_{[h]}(1,1)/N_{[h]1}, \ldots, \hat{S}_{[h]}(J,J)/N_{[h]J}\} . \]
An unbiased estimator for $\bar{Y}$ is $\breve{\bar{Y}}= \sum_{h=1}^H \omega_{[h]}\hat{\bar{Y}}_{[h]}$. Owing to the independence of treatment assignment across different strata, $N^{1/2}( \breve{\bar{Y}}- \bar{Y})$ is asymptotically Normal with mean $0$ and covariance $\sum_{h=1}^H \omega_{[h]}V_{[h]}$. A conservative variance estimator is $\breve{D}= \sum_{h=1}^H \omega_{[h]}\hat{D}_{[h]}$.

We are now positioned to make an adjustment to $X^2$ that is proper when used with the FRT in a SRE:
\begin{align}\label{e:blX2}
X^2 &=N(C \breve{\bar{Y}}-x)^\transp (C \breve{D}C^\transp)^{-1}(C \breve{\bar{Y}}-x) \nonumber \\
&=N \left( C \sum_{h=1}^H \omega_{[h]}\hat{\bar{Y}}_{[h]}-x \right)^\transp \left( \sum_{h=1}^H \omega_{[h]}C \hat{D}_{[h]}C^\transp \right)^{-1}\left( C \sum_{h=1}^H \omega_{[h]}\hat{\bar{Y}}_{[h]}-x \right)
\end{align}
The special case $h=1$ and \eqref{e:X2} agree, so the same notation $X^2$ for this statistic is logical. Besides the form of the test statistic, the FRT entails two more modifications in the case of an SRE. First, we impute the potential outcomes stratum by stratum under the sharp null hypothesis
\[ H_{0 \T{F}}(C,x_{[1]}, \ldots,x_{[H]}, \tilde{C}, \tilde{x}_{[1]}, \ldots, \tilde{x}_{[H]}): \vect{C}{\tilde{C}}Y_i^*= \vect{x_{[h]}}{\tilde{x}_{[h]}}, \T{ whenever }X_i=h. \]
Since we still aim to test \eqref{e:HNx}, the above null hypothesis must satisfy $\sum_{h=1}^H \omega_{[h]}x_{[h]}=x$. If $x=0_m$, it is natural to choose $x_{[h]}=x$ and $\tilde{x}_{[h]}=0_{J-m-1}$ for each $h$.
Under the above sharp null hypothesis, we can impute all potential outcomes: for units in stratum $h$,
\[ Y_i^*= \vecth{Y_i^*(1)}{\vdots}{Y_i^*(J)}=z_{[h]}+(Y_i^\T{obs}-z_{[h],W_i})1_J,
\T{ where }z_{[h]}= \vecth{z_{[h],1}}{\vdots}{z_{[h],J}}= \vecth{C}{\tilde{C}}{1_J^\transp}^{-1}\vecth{x_{[h]}}{\tilde{x}_{[h]}}{0}, \]
or, equivalently, $Y_i^*(j)=Y_i^\T{obs}+z_{[h],j}-z_{[h],W_i}$. Second, we ought to permute the treatment indicators within strata, independently across strata. 
Let $\Pi_{N, \T{S}} \subseteq \Pi_N$ be all such permutations from a SRE. The $p$-value is $\left( \prod_{h=1}^HN_{[h]}! \right)^{-1}\sum_{\pi \in \Pi_{N, \T{S}}}1(X_\pi^2 \geq X^2)$. 
\begin{thm}\label{t:CYx-bl}
In a SRE, suppose Assumption \ref{Asu-bl} holds. Under $H_{0 \T{N}}(C,x)$, $X^2 \ind \sum_{j=1}^ma_j \xi_j^2$, where each $a_j \in [0,1]$. If $\pi \sim \Unif( \Pi_{N, \T{S}})$, then $X_\pi^2|W \ind \chi_m^2$ a.s.
In particular, the FRT with test statistic $X^2$ can asymptotically control type I error because the condition of Proposition \ref{p:suitable} holds. 
\end{thm}

Even if the original experiment is a CRE, if a discrete covariate $X$ is available, we can condition on the number of treated and control units landing in each stratum. Then the treatment assignment is identical to a SRE.
Therefore, in a CRE, we can still permute the treatment indicators within each stratum of $X$. This plan is billed as a conditional randomization test.
\citet{zheng2008multi} and \citet{hennessy2016conditional} perceived that conditional randomization tests typically enhance the power as long as the covariates are predictive of the outcomes.

We have focused on the SRE with large strata, i.e., $N_{[h]}\to \infty$ for $h \in\ot{H}$, and $H$ is fixed. Our theory does not encapsulate SREs with many small strata, i.e., the $N_{[h]}$'s are bounded but $H \to \infty$ \citep{fogarty2018mitigating}.
Although we conjecture that similar results hold in such cases, we defer technical details to future research.

\subsection{Multiple Outcomes and Multiple Testings}
We can lengthen the reach of our framework to the case where all potential outcomes $Y_i(j) \in \bb{R}^d$ are vectors. Define $\bar{Y}(j)$ and $ \hat{\bar{Y}}(j) \in \bb{R}^d$ as before. It is convenient to gather these into long vectors
\[ 
\bar{Y}= \vecth{\bar{Y}(1)}{\vdots}{\bar{Y}(J)}\in \bb{R}^{dJ}, \qquad
\hat{\bar{Y}}= \vecth{\hat{\bar{Y}}(1)}{\vdots}{\hat{\bar{Y}}(J)}\in \bb{R}^{dJ}. 
\] 
The covariances $S(j,k)= \sum_{i=1}^N \{ Y_i(j)- \bar{Y}(j) \} \{ Y_i(k)- \bar{Y}(k) \}^\transp /(N-1)$ and $\hat{S}(j,j)= \sum_{i=1}^NW_i(j) \op{\{ Y_i^\T{obs}- \hat{\bar{Y}}(j) \}}/(N_j-1)$ are now matrices, for $j,k=\ot{J}$.
The overall covariance matrix $S \in \bb{R}^{dJ \times dJ}$ has $(j,k)$-th block $S(j,k)$. Assume $S(j,j)$ and $\hat{S}(j,j)$ are both positive definite for all realizations of $W$. 

Let $Y_i(j)_1, \ldots,Y_i(j)_d$ be the $d$ components of the potential outcomes $Y_i(j)$ for all $i$ and $j$. We wish to test the weak null hypothesis
\begin{equation}\label{e:HNx-vec}
H_{0 \T{N}}(C_1, \ldots,C_d,x_1, \ldots,x_d):
C_1 \vecth{\bar{Y}(1)_1}{\vdots}{\bar{Y}(J)_1}=x_1, \ldots,
C_d \vecth{\bar{Y}(1)_d}{\vdots}{\bar{Y}(J)_d}=x_d,
\end{equation} 
where $C_1, \ldots,C_d$ are contrast matrices that have $J$ columns and possibly varying row counts. 
We can condense notation via the Kronecker product: define
\[ 
C= \vecth{C_1 \otimes e_1^\transp}{\vdots}{C_d \otimes e_d^\transp}, \qquad x= \vecth{x_1}{\vdots}{x_d}, 
\]
where $\{ e_1, \ldots,e_d \}$ are the standard basis vectors of $\bb{R}^d$. We can then write \eqref{e:HNx-vec} in the form $H_{0 \T{N}}(C,x):C \bar{Y}=x$. It looks exactly like \eqref{e:HNx}, but $C$ cannot be an arbitrary contrast matrix.
\begin{example}
We lay out some possible contrast matrices when $J=3$ and $d=2$. The hypothesis $H_0: \bar{Y}(1)= \bar{Y}(2)= \bar{Y}(3)$ has the contrast matrix
\[ \begin{pmatrix}
1 & 0 & -1 & 0 & 0 & 0 \\
1 & 0 & 0 & 0 & -1 & 0 \\
0 & 1 & 0 & -1 & 0 & 0 \\
0 & 1 & 0 & 0 & 0 & -1
\end{pmatrix}= \vect{C_1 \otimes e_1^\transp}{C_1 \otimes e_2^\transp}, \T{ where }C_1= \begin{pmatrix}
1 & -1 & 0 \\
1 & 0 & -1 \\
\end{pmatrix}\]
Here, we test the same hypothesis entry by entry, and an equivalent contrast matrix is $C_1 \otimes I_2$.  
We can also test different hypotheses entry by entry, for instance $H_0: \bar{Y}(1)_1= \bar{Y}(2)_1$, $\bar{Y}(2)_2= \bar{Y}(3)_2$. This hypothesis has the contrast matrix
\[ \begin{pmatrix}
1 & 0 & -1 & 0 & 0 & 0 \\
0 & 0 & 0 & 1 & 0 & -1 \\
\end{pmatrix}= \vect{C_1 \otimes e_1^\transp}{C_2 \otimes e_2^\transp}, \T{ where }C_1=(1,-1,0) \T{ and }C_2= (0,1,-1). \qed \]
\end{example}
The potential outcomes framework cannot withstand comparison of different entries under different treatments, for instance $H_0: \bar{Y}(1)_1= \bar{Y}(2)_2$. Null hypotheses like these do not have a clear causal interpretation here. 
Under i.i.d. sampling, \cite{LongFriedr17} allow for a general contrast matrix $C$, and even for the length of $Y_i(j)$ to depend on treatment $j$. We constrain the contrast matrices $C$ that we accept, as we have just detailed.

Under i.i.d. sampling and vector potential outcomes, \cite{ManovaRomano16} address the two-sample problem with permutation tests. \citet{srivastava2013tests}, \citet{Kon15} and \citet{MatsFriedr18} test general linear hypotheses with bootstrap methods.
We will use the FRT for \eqref{e:HNx-vec}. It is not a sharp null hypothesis, so we concoct one:
\[ H_{0 \T{F}}: \vect{C_1}{\tilde{C}_1}\vecth{Y_i(1)_1}{\vdots}{Y_i(J)_1}= \vect{x_1}{\tilde{x}_1}, \ldots,
\vect{C_d}{\tilde{C}_d}\vecth{Y_i(1)_d}{\vdots}{Y_i(J)_d}= \vect{x_d}{\tilde{x}_d}, \T{ for }i=\ot{N}, \]
where the matrices $(C_1^\transp, \tilde{C}_1^\transp,1_J)$ through $(C_d^\transp, \tilde{C}_d^\transp,1_J)$ are invertible. We construct the $\tilde{C}$'s and $\tilde{x}$'s for each component of the outcome in the same way as the scalar case.
In the hypothesis $H_{0 \T{F}}$, our notation does not reflect its dependence on the $C$'s, $\tilde{C}$'s, $x$'s and $\tilde{x}$'s. We impute potential outcomes as if $H_{0 \T{F}}$ were the reality. For the first component:
\begin{equation}\label{e:impute-vec}
\vecth{Y_i^*(1)_1}{\vdots}{Y_i^*(J)_1}=z_1+(Y_{i,1}^\T{obs}-z_{1W_i})1_J, \T{ where }
z_1= \vecth{z_{11}}{\vdots}{z_{1J}}= \vecth{C_1}{\tilde{C}_1}{1_J^\transp}^{-1}\vecth{x_1}{\tilde{x}_1}{0}
\end{equation}
and similarly for the second through the $d$-th entries, replacing all subscripts $1$ by $2, \ldots,d$.

For vector potential outcomes, we tweak $X^2$ in \eqref{e:X2}:
\[ X^2 =N(C \hat{\bar{Y}}-x)^\transp (C \hat{D}C^\transp )^{-1}(C \hat{\bar{Y}}-x), \]
where the block diagonal matrix $\hat{D}=N \cdot \diag\{ \hat{S}(1,1)/N_1, \ldots,\hat{S}(J,J)/N_J \}$ is an asymptotically conservative estimator of $N \cdot \Cov( \hat{\bar{Y}})$. This is in sync with \eqref{e:Dob}. 
The FRT with $X^2$ can control the asymptotic type I error under \eqref{e:HNx-vec}. We first give the asymptotic requirements and then adapt Theorem \ref{t:CYx} to the vector case. Let $| \cdot |$ be the Euclidean norm, which reduces to the usual absolute value for scalars.
\begin{assume}\label{AsuA-vec}
The sequence $(N_j/N)$ converges to $p_j \in (0,1)$ for all $j=\ot{J}$. The sequences $( \bar{Y}_N)$ and $(S_N)$ converge to $\bar{Y}_\infty$ and $S_\infty$, where $| \bar{Y}_\infty |< \infty$, $S_\infty$ is positive semi-definite, and $S_\infty(j,j)$ is positive definite for all $j=\ot{J}$.
Further, $\lmt{N}\max_{j=\ot{J}}\max_{i=\ot{N}}|Y_i(j)- \bar{Y}(j)|^2/N=0$.
\end{assume}
\begin{assume}\label{AsuB-vec}
Same as Assumption \ref{AsuA-vec} except the last sentence is replaced by: Further, there exists an $L< \infty$ such that $\sum_{i=1}^N|Y_i(j)- \bar{Y}(j)|^4/N \leq L$ for all $j=\ot{J}$ and $N \geq (d+1)J$.
\end{assume}
\begin{prop}\label{p:BimpA-vec}
Assumption \ref{AsuB-vec} implies Assumption \ref{AsuA-vec}.
\end{prop}
\begin{thm}\label{t:MPO}
If Assumption \ref{AsuA-vec} holds, then under $H_{0 \T{N}}(C,x)$, $X^2 \ind \sum_{j=1}^ma_j \xi_j^2$, where each $a_j \in [0,1]$.
If Assumption \ref{AsuB-vec} holds and $\pi \sim \Unif( \Pi_N)$, then $X_\pi^2|W \ind \chi_m^2$ a.s. In particular, the FRT with test statistic $X^2$ can asymptotically control type I error a.s.
\end{thm}
Theorem \ref{t:MPO} puts in place a foundation for a single FRT for multiple outcomes. As done in \citet[][Section 4]{ManovaRomano16}, we can join Theorem \ref{t:MPO} and the closure procedure for multiple testings. We omit the details.
  
To conduct the FRT with $X^2$ at all, we require all realizations of $\hat{S}(j,j)$ to be invertible, for which it is necessary that $N_j \geq d+1$.
\cite{MatsFriedr18} instead tried $\tilde{X}^2 = N(C \hat{\bar{Y}}-x)^\transp (C \tilde{D}C^\transp )^{-1}(C \hat{\bar{Y}}-x)$ with a bootstrap, where $\tilde{D}$ is a diagonal matrix whose main diagonal is the same as $\hat{D}$. 
However, $\tilde{X}^2$ is not proper for the FRT because the asymptotic distribution of $\tilde{X}_\pi^2|W$ is not pivotal. So it is flawed for the same reason the Box type statistic $B$ in \eqref{e:BoxStat} is. We reserve FRTs with $d \to \infty$ for future research.

\section{Simulations}\label{s:sim}

\subsection{Type I Error Rates of FRTs with Different Statistics}

We perceive from previous sections that $X^2$ is proper, but $B$ and $F$ are not. As a complement to this asymptotic fact, simulations reveal their finite sample behavior. To drive this point, we repeat the simulations with varying sample sizes.
All the test statistics we brought up had other specific purposes in the literature. Thus, the simulations also serve to compare their efficacy with the FRT for testing weak null hypotheses. 

\subsubsection{Simulation Setup}
We decided on the ANOVA with $J=3$ and the $2^2$ Factorial  with $J=4$ setup, which we refer to as ``ANOVA'' and ``Factorial'' for short. The null hypotheses being tested, written in the form of \eqref{e:HNx}, are
\[ H_{0 \T{N}}: \begin{pmatrix}
1 & -1 & 0 \\
1 & 0 & -1 \end{pmatrix}\bar{Y} = 0\T{ for ANOVA, and }H_{0 \T{N}}: \begin{pmatrix}
-1 & -1 & 1 & 1 \\
-1 & 1 & -1 & 1 \end{pmatrix}\bar{Y} = 0 \T{ for } \T{Factorial.}\]
In words, the former tests for no effects of any treatments on average. The latter tests for no main effects of either of the two factors on average. Both setups shall have a balanced design $N_j=N/J$ for all $j$. We then gain from Proposition \ref{p:BeqF} that $B=F$.
Thus, a comparison of $X^2$ and $B$ suffices. In all cases, we compel $\bar{Y}(1)= \cdots= \bar{Y}(J)=0$, so the weak null hypothesis of no treatment effects on average holds.
We also compel force the covariance structure $S=uu^\transp$ on the potential outcomes. For the ANOVA case, $u^\transp =(u_1,u_2,u_3)=(1,2,3)$, and for the Factorial case, $u^\transp =(u_1,u_2,u_3,u_4)=(3,1,1,3)$.
We deliberately avoid any sharp null hypothesis being true by design. Otherwise, all test statistics would have correct type I error control.

Explicitly, we first generate $Y_i(1) \iid \cl{N}(0,1)$ for $i=\ot{N}$, center them, and scale them according to $Y_i(j)=u_jY_i(1)$. For the hypothesis test itself, we simulate 10000 different realizations of the observed outcomes. 
For each set of $(W_i , Y_i^\T{obs})_{i=1}^N$, we run the FRT with both $X^2$ and $B$, calculating $p$-values from 2500 permutations.

For these potential outcomes, we compute the eigenvalues in Theorems \ref{t:CYx} and \ref{t:Brunner} to derive that the asymptotic sampling distributions of $X^2$ and $2B$ under $H_{0 \T{N}}$ are
\begin{align}
X^2 \ind \xi_1^2+0.758 \xi_2^2, & \qquad 2B \ind 1.423 \xi_1^2+0.434 \xi_2^2, & \T{(ANOVA)}, \label{e:SimDistr}\\
X^2 \ind \xi_1^2+ \xi_2^2 \eqd \chi^2_2, & \qquad 2B \ind 1.8 \xi_1^2+0.2 \xi_2^2, & \T{(Factorial)}; \nonumber
\end{align}
their randomization distributions are both asymptotically $\chi_2^2$ in both the ANOVA and factorial designs. This provides an illustrative and simple numerical example of our main results. Each weight for $X^2$ is at most $1$, while the weights for $2B$ are only at most $1$ on average. 
In the Factorial case, the FRT with $X^2$ is actually asymptotically exact because both the sampling and randomization distributions of $X^2$ approach $\chi^2_2$.

We can naturally broaden the simulations just performed to SREs. We keep the ANOVA and Factorial setup, but now incorporate a SRE with $H=2$ strata. Remember that this means the observed data come from running a CRE within each stratum separately.
The first stratum of potential outcomes shall be identical to those of the ANOVA simulation above. The second stratum shall be identical to the first, except a unit constant is added to all its potential outcomes.
This between stratum effect merits a SRE analysis. We proceed with the $X^2$ statistic in \eqref{e:blX2}, and only permute data within each stratum when obtaining $p$-values.

The textbook suggestion \cite{morris2010design} for testing the our null hypotheses in the SRE case involves the $F$ statistic from a linear regression of the observed response on stratum and treatment indicators, i.e., $J+H$ predictors.
Although \cite{morris2010design} has reiterated the usual OLS assumptions that justify the $F$ test, practitioners do not always check them. We therefore would like to compare $X^2$ and $F$ in this SRE setting. 
From Theorem \ref{t:CYx-bl}, we know $X^2$ in \eqref{e:blX2} has the same asymptotic behavior as listed in \eqref{e:SimDistr}. By intuition from \cite{Lin13}, we anticipate that $2F$ also has the same asymptotic behavior as before.

In all four settings we have put forth, we also fix three different sample size settings to pinpoint the rate that asymptotics take effect.  

\subsubsection{Results}

\begin{figure}[t]  
\centering
\includegraphics{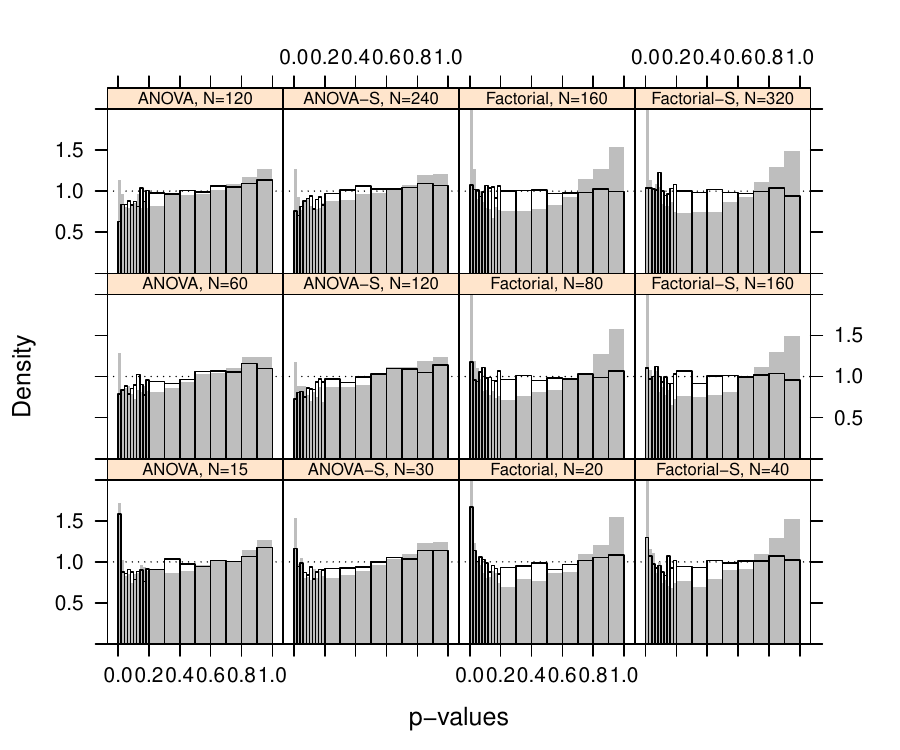}
\caption{Histograms of FRT $p$-values under various settings and sample sizes, with ``S'' indicating the stratified cases. Gray bars indicate $p$-values from a $F$ statistic, while transparent bars indicate $p$-values from the $X^2$ statistic. We display smaller $p$-values with a finer resolution because most hypothesis tests are conducted at levels close to 0. A dashed line indicating the $\Unif(0,1)$ density is added for reference purposes.}
\label{f:VaryingN}
\end{figure}

Figure \ref{f:VaryingN} contains the simulation results. For each setting and sample size, we plot histograms of $p$-values from the FRT with $X^2$ and $B$ or $F$. In all histograms, the left-most bin of $p$-values ranging from 0 to 2\% is most informative.
For a successful control of type I error, the density of $p$-values here should not surpass 1 by much. From the bottom row of Figure \ref{f:VaryingN}, $N_1$ or $N_{[1]1}=5$ (bottom row) is evidently far from the asymptotic regime.
When $N_1$ or $N_{[1]1}=20$ (middle row), it appears that we move much closer to the expected behavior dictated by asymptotics.
This is because, when these counts are 40 (top row), the histograms do not change much from the row below. That is, the first and second rows have a similar pattern.
The similarity of the SRE histograms to the corresponding CRE ones buttresses our intuition that $X^2$ and $F$ have similar distributions as their ``unstratified'' twins for our simulated potential outcomes.

It is also confirmed that the FRT with $B$ or $F$ fails to control type I error at small $p$-values for any sample size. We recollect from our theory that heteroscedasticity hampers its suitability.
We have elected to balance the designs, so that it surfaces that, when $J>2$, balanced designs do not guarantee the suitability of $B$ or $F$ as they do in treatment-control experiments (refer to Corollary \ref{c:tc}).
Of course, forgoing balanced designs can cause both $B$ and $F$ to fail more seriously. \cite{DD18} compare $X^2$ and $F$ in such cases through extensive simulation.

\subsection{Confidence Regions}
Our next simulation constructs confidence regions alluded to by Corollary \ref{c:CI-FRT}. At the same time, we seize the opportunity to compare the FRT and $\chi^2$ approximations that are both asymptotically valid by Theorem \ref{t:CYx}.
We decided on a balanced $2^2$ factorial design ($K=2$, $J=2^2=4$) where $N_j=10$ for $j=\ot{4}$. We seek to infer the main effects $\tau_1$, $\tau_2$, both individually and jointly.  
Take $Y_i(j) \iid U^2-1/3$ where $U \sim \Unif(0,1)$, and center so that each $\bar{Y}(j)=0$. This way, the true parameter values are $\tau_1= \tau_2=0$, but takeaways of this simulation generalize to arbitrary $\tau_1, \tau_2$.  Next, multiply each $Y_i$ by the same matrix 
\[ \left( \begin{matrix}
2 & 1 & 3/2 & 1 \\
0 & \sqrt{5} & \sqrt{5}/2 & 2/ \sqrt{5}\\
0 & 0 & 3/ \sqrt{2} & 1/ \sqrt{2}\\
0 & 0 & 0 & \sqrt{3.7}
\end{matrix}\right) \]
to inject correlation into the potential outcomes.

We assign treatments to units according to the CRE, and construct the confidence regions by means of a single realization of observed outcomes. The set $\T{CR}_\alpha$ in \eqref{e:CIapprox} is a means to compute an asymptotic confidence region for $\tau_1$, $\tau_2$.
After finding it, we spread a grid of points centered at $\hat{\tau}_1$, $\hat{\tau}_2$ that comfortably envelops this asymptotic region.
At each point ($x_1,x_2$) of this grid, we run the FRT with $X^2$ to test $\tau_1=x_1$, $\tau_2=x_2$, both individually and jointly. We induct the point into our confidence region if and only if the $p$-value exceeds $\alpha=0.05$.

Figure \ref{f:1D-CI} shows the results for the marginal hypothesis tests. The behavior is very regular: the $p$-value crests near $\hat{\tau}_1$ or $\hat{\tau}_2$, and decays monotonically to the left and right. The FRT and $\chi^2$ approximation confidence intervals are nearly indistinguishable.
Figure \ref{f:ConfSet} shows the result for the joint test. The left graph shows the FRT confidence region is again close to its asymptotic approximation, but not as close as in the 1D case.
In particular, the former is noticeably larger. The right graph explains this by exposing that the $p$-values calculated from the FRT tend to be larger than those from the $\chi^2$ approximation.

Due to the duality between hypothesis testing and confidence regions, the empirical coverage of our confidence region is the proportion of time it includes $\tau_1= \tau_2=0$ over all realizations of the observed data.
From the simulations in the previous section, which deals with the false rejection rate of the FRT, we expect this proportion to be at least 0.95.
The closeness of the confidence regions to their asymptotic approximations suggests our results generalize to other realizations of the observed data. That is, those confidence regions will be centered at $( \hat{\tau}_1, \hat{\tau}_2)$, but have similar shape.

\begin{figure}[h]
        \begin{subfigure}{\textwidth}
\includegraphics[width=0.48\linewidth]{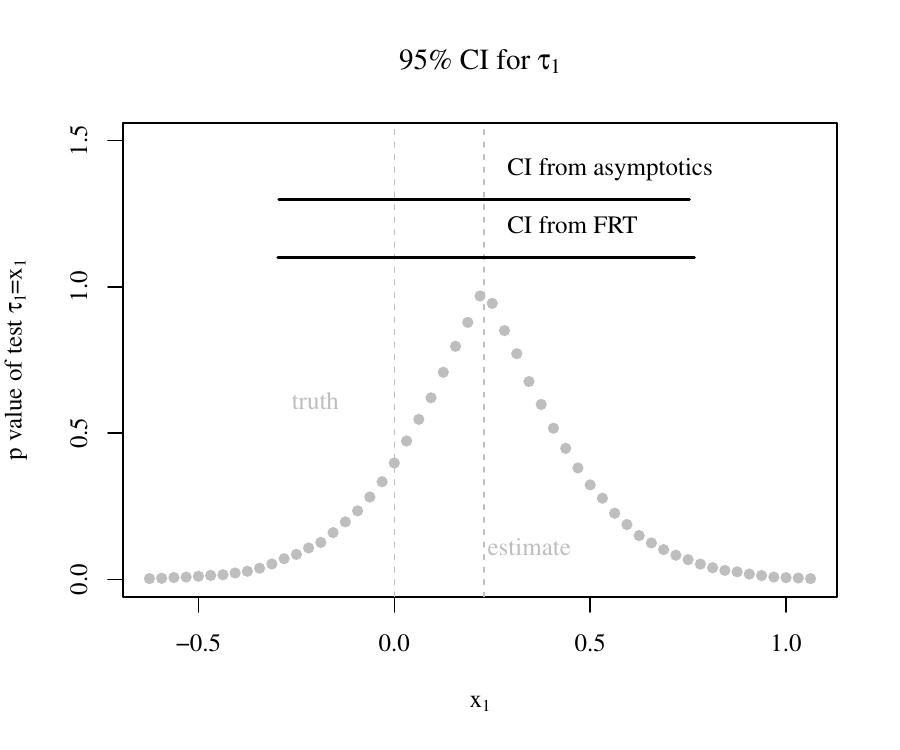}\hfill
\includegraphics[width=0.48\linewidth]{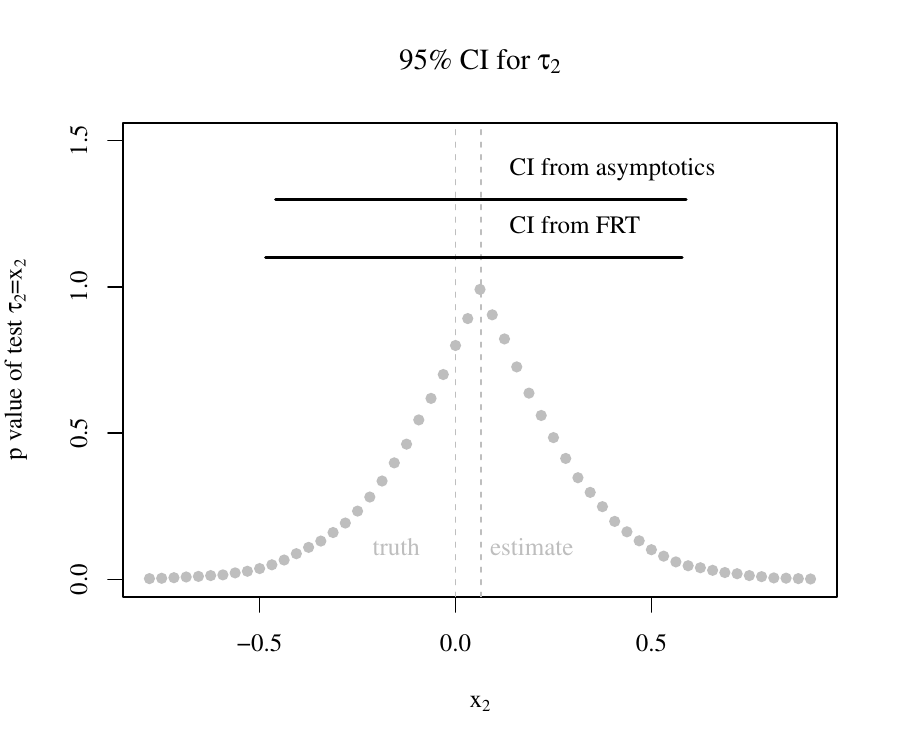}
\caption{For $\tau_1$ and $\tau_2$ individually, the FRT and asymptotic approximation give nearly identical confidence intervals (CI). For the second main effect, the FRT confidence interval is shifted due to the discrete resolution.}  
\label{f:1D-CI}
        \end{subfigure}%

        \begin{subfigure}{\textwidth}
\includegraphics[width=0.48\linewidth]{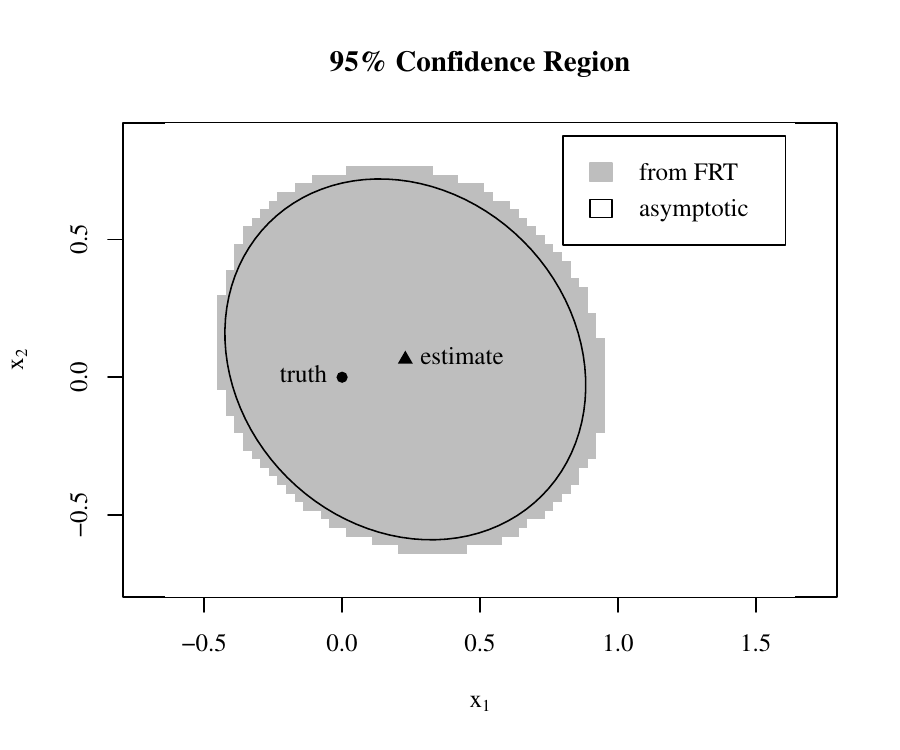}\hfill
\includegraphics[width=0.48\linewidth]{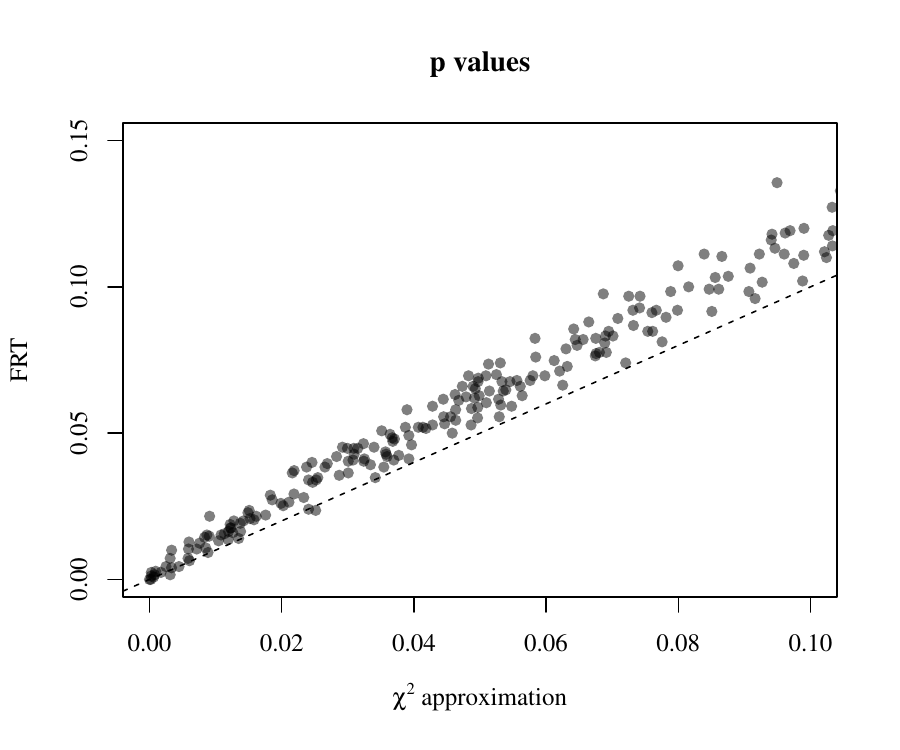}
\caption{The left graph shows the FRT confidence region is again close to its asymptotic approximation, but the former is noticeably larger. The right graph is a scatter plot of $p$-values from testing $\tau_1= \tau_2=0$ repeatedly from the original set of potential outcomes, zooming in on the region where they are less than 0.1.}
\label{f:ConfSet}
        \end{subfigure}%
        \caption{Simulation for confidence regions}\label{f:confidence}
\end{figure}

\section{Applications}\label{s:appl}
We now try out our method on practical datasets, under a variety of possible weak null hypotheses. Our goal is not to do complete data analyses. We do not delve into issues of multiple comparisons. We pretend each null hypothesis is tested in isolation.

\subsection{Financial Incentives for Exercise}
\cite{Charness09} were interested in whether financial incentives caused college students to exercise more. They randomly assigned 40 students each to one of three possible treatments: no financial incentive (control), a small one, or a large one.
We henceforth index these groups by $j=1,2,3$, respectively. Then $N_1=N_2=N_3=40$. For each student, the response was the average number of weekly gym visits after the study minus that before the study. Let $Y_i(j)$ denote this quantity for the $i$-th student, if s/he received treatment $j$.
Many students had $Y_i^\T{obs}=0$. This would be troublesome for the FRT with $X^2$ if, after a certain permutation, all permuted observations in a group were 0. To preclude this, we added a minuscule amount of random noise to all the $Y_i^\T{obs}$. 
For this dataset, the sample means are $-0.029, 0.054, 0.640$, and the sample variances are $0.152,0.386,1.489$, for groups $j=1,2,3$, respectively.
Mere inspection of these numbers posits that a large financial incentive has a positive effect while a small one does not. It is also apparent that the data are heteroscedastic.

We test these four hypotheses at level 1\%: whether the two magnitudes of financial incentives have any effect on average, whether financial incentives have any effect ignoring the division between large and small, whether financial incentives have any effect, and whether small financial incentives have any effect.
In symbols, these are $2 \bar{Y}(1)= \bar{Y}(2)+ \bar{Y}(3)$, $\bar{Y}(1)= \bar{Y}(2,3)$ (here we collapse treatment levels $j=2,3$ to one), $\bar{Y}(1)= \bar{Y}(2)= \bar{Y}(3)$, and $\bar{Y}(1)= \bar{Y}(2)$ (here we ignore the $W_i=3$ observations), respectively.

We use the $X^2$ and $F$ statistics, and get $p$-values both by the FRT and the $\chi^2$ (or $F$) approximation. As we brought up earlier, $p$-values from FRTs are also finite-sample exact for testing Fisher's sharp null hypothesis. Consult Table \ref{tb::charness09} for the results.
The class of hypothesis test (FRT and $\chi^2$ (or $F$) approximation) holds little sway. It seems, for $X^2$, the FRT is slightly more conservative. For $F$, the FRT is slightly less conservative.

\begin{table}
\centering 
\caption{Analyzing \cite{Charness09}'s data with $p$-values as percents. We calculate the FRT $p$-values using $10^4$ Monte Carlo simulations and the asymptotic $p$-values based on $\chi^2$ or $F$ approximations.}\label{tb::charness09}
\begin{tabular}{|c||c|c||c|c|}\hline
Hypothesis &   $X^2 \ind \chi_m^2$  & FRT using $X^2$ &  $F \ind F_{m,N-J}$  & FRT using $F$ \\ \hline
$2 \bar{Y}(1)= \bar{Y}(2)+ \bar{Y}(3)$ &  0.25 & 0.27 &   1.97 & 1.59 \\ \hline
$\bar{Y}(1)= \bar{Y}(2)= \bar{Y}(3)$ &   0.42 & 0.49 &  0.06 & 0.01 \\ \hline
$\bar{Y}(1)= \bar{Y}(2,3)$ &   0.34 & 0.49   & 2.45 & 2.34 \\ \hline
$\bar{Y}(1)= \bar{Y}(2)$ &   47.15 & 47.93 &   47.37 & 47.93 \\ \hline
\end{tabular}
\end{table}
Testing the first two hypotheses, financial incentives have a statistically significant impact on gym attendance. Guided by Theorems \ref{t:CYx} and \ref{t:LSF}, we should trust the $p$-values from $X^2$ more than those from $F$. The latter statistic seems to have overly conservative behavior for this dataset. 
Testing the third hypothesis suggests that the treated group ($j=2$ or $3$) has different behavior from the control in a statistically significant way. 

Seeing evidence that financial incentives might be helpful, we test the fourth hypothesis only comparing the control and small incentive groups, and get insignificant $p$-values. Note, in this case, $X^2=F$  by Corollary \ref{c:tc}, thanks to the balanced design.
To wrap up, we concur with the findings of \cite{Charness09}, that large financial incentives seem to induce people to visit the gym more often, but not small ones.

\subsection{A $2^2$ Factorial Experiment for Grades}
We now undertake a similar analysis as in the previous section on another dataset. \cite{AngristStar09} wondered whether academic support services and/or financial incentives caused college students to improve their grades.
Their data consisted of student grades for a certain semester on a 100 point scale. In that semester, students were either in a control group, offered a fellowship, offered services, or both.
We thus have a $2^2$ factorial experiment, and henceforth index these treatment groups by $j=1,2,3,4$, respectively. As opposed to the allocation in the previous section, this one is imbalanced: $(N_1,N_2,N_3,N_4)=(854,219,212,119)$. 
The sample means are $63.9, 65.8, 64.1, 66.1$, and the sample variances are $145,124,160,114$, for groups $j=1,2,3,4$, respectively. By eye, there is less heteroscedasticity, and the sample means are less markedly off from each other than those of the previous section.

We test the following five hypotheses at level 1\%: financial services have no effect, services have no effect, neither has an effect, no interactions, and that all group means are the same.
In symbols, these are $\bar{Y}(1)+ \bar{Y}(2)= \bar{Y}(3)+ \bar{Y}(4)$, $\bar{Y}(1)+ \bar{Y}(3)= \bar{Y}(2)+ \bar{Y}(4)$, both of the previous two, $\bar{Y}(1)+ \bar{Y}(4)= \bar{Y}(2)+ \bar{Y}(3)$, and $\bar{Y}(1)= \bar{Y}(2)= \bar{Y}(3)= \bar{Y}(4)$.

We again use the $X^2$ and $F$ statistics, and get $p$-values both by the FRT and the $\chi^2$ (or $F$) approximation. As we discussed earlier, $p$-values from FRTs are also exact for testing Fisher's sharp null hypothesis.
Consult Table \ref{tb::angriststar09} for the results. The class of hypothesis test again holds little sway. The FRT seems as a whole slightly more conservative, though there are a few exceptions. 
\begin{table}
\centering
\caption{Analyzing \cite{AngristStar09}'s data with $p$-values as percents. We calculate the FRT $p$-values using $10^4$ Monte Carlo simulations and the asymptotic $p$-values based on $\chi^2$ or $F$ approximations.}\label{tb::angriststar09}
\begin{tabular}{|c||c|c||c|c|}\hline
Hypothesis &   $X^2\ind \chi_m^2$   & FRT using $X^2$ &   $F\ind F_{m,N-J}$  & FRT using $F$ \\ \hline  
No effect from services &   72.84 & 72.34 &   73.92 & 73.58 \\ \hline
No effect from incentives &   1.19 & 1.43 &   1.60 & 1.80 \\ \hline
No effects from either &  3.65 & 3.99 &   5.26 & 5.28 \\ \hline
No interaction &   99.53 & 99.47 &   99.55 & 99.5 \\ \hline
$\bar{Y}(1)= \bar{Y}(2)= \bar{Y}(3)= \bar{Y}(4)$ &   3.88 & 4.31 &   5.85 & 5.71 \\ \hline
\end{tabular}
\end{table}
We cannot reject any of these null hypotheses at level 1\%. From the second and fourth hypotheses, the data do not seem to suggest services have any effect, or that there is a non-additive effect from combining incentives and services.
We do, however, almost reject the hypothesis of no effect from incentives alone: the $p$-values are just over 1\%.

Our finding that the effect of incentives is more significant than the effect of others conforms with the conclusions of \cite{AngristStar09}. They went on to conduct subgroup analysis, and discovered that the observed effects on grades come nearly exclusively from female students.

\section{Discussion}\label{s:disc}
We have proposed a strategy for using the FRT to test a weak null hypothesis. It imputes the missing potential outcomes under a compatible sharp null hypothesis, and then uses the studentized statistic in the FRT.
It furthers the current literature in two directions. First, it complements the tests centered on asymptotic distributions. Our FRT is also finite-sample exact under the sharp null hypothesis.
Second, it guides the choice of test statistic for the sharp null hypothesis. Although the finite-sample exactness property of the FRT holds for any test statistic, the $p$-values are sensitive to this choice.
For example, all the $p$-values in Tables \ref{tb::charness09} and \ref{tb::angriststar09} are valid for Fisher's sharp null hypothesis. Unfortunately, these $p$-values range above and below the nominal significance level.
This can be confusing in practice. Therefore, we cannot overstate the crucial role of weak null hypotheses and studentized statistics. Our FRTs can control asymptotic type I error under weak null hypotheses and have power under corresponding alternative hypotheses.

Our theory ignores covariates. The analysis of covariance is a classical topic \citep{Fisher35} and still attracts attention \citep{Lin13, lu2016covariate, fogarty2018regression, fogarty2018mitigating, middleton2018unified}.
\citet{LassoTE16} and \citet{lei2018regression} widened it to the case where the number of covariates grows with the sample size.
\citet{tukey1993tightening} and \citet{CovAdjRosen02} discussed strategies for testing sharp null hypotheses. It is important to extend the theory to test weak null hypotheses with covariate adjustment, plus to the case with high dimensional covariates. We leave this to future work.

We have focused on completely randomized factorial experiments and extended the theory to stratified and clustered experiments. We conjecture that the strategy is also applicable for experiments with general treatment assignment mechanisms \citep{mukerjee2018using}.
\citet{fogarty2016studentized} also used the idea of studentization in sensitivity analysis of matched observational studies.

\bibliographystyle{plainnat}
\bibliography{refs}

\clearpage 
\pagenumbering{arabic} 
\renewcommand*{\thepage}{A \arabic{page}} 
\begin{center}
{\bf \Large   
\begin{spacing}{1.25}
Supplementary Material for
``Randomization Tests for Weak Null Hypotheses in Randomized Experiments'' 
\end{spacing}
}
\end{center}

\setcounter{equation}{0}
\setcounter{section}{0}
\setcounter{figure}{0}
\setcounter{example}{0}
\setcounter{prop}{0}
\setcounter{cor}{0}
\setcounter{thm}{0}
\setcounter{table}{0}

\renewcommand {\theprop} {A\arabic{prop}}
\renewcommand {\theexample} {A\arabic{example}}
\renewcommand {\thefigure} {A\arabic{figure}}
\renewcommand {\thetable} {A\arabic{table}}
\renewcommand {\theequation} {A\arabic{equation}}
\renewcommand {\thelem} {A\arabic{lem}}
\renewcommand {\thesection} {A\arabic{section}}
\renewcommand {\thethm} {A\arabic{thm}}
\renewcommand {\thecor} {A\arabic{corollary}}
\renewcommand {\theassume} {A\arabic{assume}}

\allowdisplaybreaks

Let $| \cdot |$ be the absolute value of a scalar or the Euclidean norm of a vector. Let $\| \cdot \|_F$ be the Frobenius norm of a matrix. For $A,B \in \bb{R}^{m \times n}$, let $A*B$ be the component-wise product of $A$ and $B$: $(A*B)_{ij}=A_{ij}B_{ij}$.
Let $\max_i$, $\max_j$, and $\max_{i,j}$ denote the maximums over $\{ i=\ot{n}\}$, $\{ j=\ot{J}\}$, and both. Let $a \vee b= \max(a,b)$ be the maximum value of $a$ and $b$.

Appendix \ref{a::lemmas} gives several useful lemmas and their proofs. 
Appendix \ref{s:pf} gives the proofs of the main theorems.
Appendix \ref{a::proofothers} gives the proofs of other corollaries and propositions. 

\section{Lemmas}\label{a::lemmas}
\begin{lem}\label{l:DD18-LA}
\begin{enumerate}[(i)]
\item If $X \sim \cl{N}(0_J,A)$, then $X^\transp BX \eqd \sum_{j=1}^J \lambda_j(AB) \xi_j^2$. If $A$ is a projection matrix, then each $\lambda_j(AB) \leq \lambda_1(B)$.
\item If $A,B \succeq 0$ and $B$ is a correlation matrix, then $\lambda_1(A*B) \leq \lambda_1(A)$.
\item If $X_n \ind \cl{N}(0_m,A)$, and $B_n \inP B \succ 0$, then $X_n^\transp B_n^{-1}X_n \ind \sum_{j=1}^m \lambda_j(AB^{-1}) \xi_j^2$. If $B \succeq A$, then each $\lambda_j(AB^{-1}) \in[0,1]$.
\end{enumerate}
\end{lem}
\begin{proof}
(i) and (ii) come from \cite{DD18}. We prove (iii). The Continuous Mapping Theorem implies $B_n^{-1}\inP B^{-1}$, and Slutsky's Theorem then implies $X_n^\transp B_n^{-1}X_n \ind X^\transp B^{-1}X$.
By (i), $X^\transp B^{-1}X \eqd \sum_{j=1}^m \lambda_j(AB^{-1}) \xi_j^2$. If $B \succeq A$, then each $\lambda_j(AB^{-1}) \in[0,1]$.
\end{proof}
\begin{lem}\label{l:Massart}
A finite population $(Y_1, \ldots,Y_N)$ has mean $\bar{Y}_N$ and variance $S_N = (N-1)^{-1}\sum_{i=1}^N(Y_i- \bar{Y}_N)^2$. Let $\cl{A}\subseteq \{ \ot{N}\}$ be a simple random sample of size $N_1$, and $\hat{\bar{Y}}_N=N_1^{-1}\sum_{i \in \cl{A}}Y_i$. Then for $t \geq 0$,
\[ \bb{P}( \hat{\bar{Y}}_N- \bar{Y}_N \geq t) \vee \bb{P}( \hat{\bar{Y}}_N- \bar{Y}_N \leq -t) \leq
\exp \left\{ - \frac{Np_{N,1}^2t^2}{C_NS_N}\right\} \leq \exp \left\{ - \frac{Np_{N,1}^2t^2}{CS_N}\right\}, \]
where $p_{N,1}=N_1/N$, $C_N= \left[ 1+ \min\left\{ 1,9p_{N,1}^2,9(1-p_{N,1})^2 \right\}/70 \right]^2$ and $C=(71/70)^2$.
\end{lem}
\begin{proof}
\cite{LassoTE16} prove the first inequality. The second follows from $C_N\leq C.$
\end{proof}
Lemma \ref{l:Massart} is crucial for our proof of almost sure convergence for sampling without replacement, as we are about to see.
\begin{lem}\label{l:asconv}
Let $\big( \{ Y_{N,i}:i=\ot{N}\} \big) $ be a sequence of populations with means $( \bar{Y}_N)$ and variances $(S_N)$. Suppose we take a simple random sample from each population of size $N_1 \geq 2$ with sample mean $\hat{\bar{Y}}_N$ and variance $\hat{S}_N$. Assume $\lmt{N}N_1/N=p_1>0$.
\begin{enumerate}[(i)]
\item If the sequence $(S_N)$ is bounded above by $S_\T{max}< \infty$, then $| \hat{\bar{Y}}_N- \bar{Y}_N| \as 0$. If we also have $\lmt{N}\bar{Y}_N= \bar{Y}_\infty$, then $\hat{\bar{Y}}_N \as \bar{Y}_\infty$. Assumption \ref{AsuA} implies these results.
\item If there is $L< \infty$ such that $\sum_{i=1}^N(Y_{N,i}- \bar{Y}_N)^4/N \leq L$ for all $N$, then $| \hat{S}_N-S_N| \as 0$. If we also have $\lmt{N}S_N=S_\infty$, then $\hat{S}_N \as S_\infty$. Assumption \ref{AsuB} implies these results.
\end{enumerate}
\end{lem}
\begin{proof}
(i) Because $p_{N,1}=N_1/N \to p_1$, we can pick a positive integer $N^*$ such that $N \geq N^*$ implies $p_{N,1}>p_1/2$. Then by Lemma \ref{l:Massart}, there is a universal constant $C \in (0, \infty)$, independent of $N$, such that, for $N \geq N^*$ and $t \geq 0$,
\[ \begin{split}
\bb{P}(| \hat{\bar{Y}}_N- \bar{Y}_N| \geq t) \leq 2 \exp \left\{ - \frac{Np_{N,1}^2}{CS_N}t^2 \right\} \leq 2 \exp \left\{ - \frac{p_1^2}{4CS_\T{max}}Nt^2 \right\} \\
\Longrightarrow \sum_{N \geq N^*}\bb{P}(| \hat{\bar{Y}}_N- \bar{Y}_N| \geq t) \leq 2 \sum_{N \geq N^*}\exp \left\{ - \frac{p_1^2}{4CS_\T{max}}Nt^2 \right\} < \infty.
\end{split}\]
By the Borel--Cantelli Lemma, $| \hat{\bar{Y}}_N- \bar{Y}_N| \as 0$.

(ii) First, by the Cauchy--Schwarz Inequality, we have that for all $N$
\[ S_N= \frac{1}{N-1}\sum_{i=1}^N(Y_{N,i}- \bar{Y}_N)^2 \leq \frac{N^{1/2}}{N-1}\left\{ \frac{1}{N}\sum_{i=1}^N(Y_{N,i}- \bar{Y}_N)^4 \right\}^{1/2}\leq \frac{N}{N-1}L^{1/2}, \]
which is bounded above as $N \to \infty$, so by (i), $| \hat{\bar{Y}}_N- \bar{Y}_N| \as 0$.

Second, let $W_{N,i}$ be the indicator for $Y_i$ being in the simple random sample. Define as an intermediate quantity $\tilde{S}_N= \sum_{i=1}^NW_{N,i}(Y_{N,i}- \bar{Y}_N)^2/(N_1-1)$, which differs from $\hat{S}_N$ by an almost surely zero quantity as $N \to \infty$:
\[ \begin{split}
\hat{S}_N- \tilde{S}_N=& \frac{1}{N_1-1}\sum_{i=1}^NW_{N,i} \big\{ (Y_{N,i}- \hat{\bar{Y}}_N)^2-(Y_{N,i}- \bar{Y}_N)^2 \big\} \\
=&\frac{1}{N_1-1}\sum_{i=1}^NW_{N,i} ( \bar{Y}_N- \hat{\bar{Y}}_N)( 2 Y_{N,i} - \hat{\bar{Y}}_N - \bar{Y}_N ) \\
=& \frac{1}{N_1-1} \left\{ 2( \bar{Y}_N- \hat{\bar{Y}}_N) \sum_{i=1}^NW_{N,i}Y_{N,i}+N_1( \hat{\bar{Y}}_N^2- \bar{Y}_N^2) \right\} \\
=& \frac{N_1}{N_1-1} \big\{ 2( \bar{Y}_N- \hat{\bar{Y}}_N) \hat{\bar{Y}}_N+ \hat{\bar{Y}}_N^2- \bar{Y}_N^2 \big\} \\
=& \frac{-N_1}{N_1-1}( \hat{\bar{Y}}_N- \bar{Y}_N)^2 \as 0 .
\end{split}\]
Third, we note that the variance of $\{ (Y_{N,i}- \bar{Y}_N)^2 \}_{i=1}^N$ is bounded above for all $N$:
\[ \Var \left[ \{ (Y_{N,i}- \bar{Y}_N)^2 \}_{i=1}^N \right] \leq \frac{1}{N-1}\sum_{i=1}^N(Y_{N,i}- \bar{Y}_N)^4 \leq \frac{N}{N-1}L. \]
So by (i), $\big| \sum_{i=1}^NW_{N,i}(Y_{N,i}- \bar{Y}_N)^2/N_1- \sum_{i=1}^N(Y_{N,i}- \bar{Y}_N)^2/N \big| \as 0$, and therefore
\[ \begin{split}
| \tilde{S}_N-S_N|=& \left| \frac{N_1}{N_1-1}\frac{1}{N_1}\sum_{i=1}^NW_{N,i}(Y_{N,i}- \bar{Y}_N)^2- \frac{N_1}{N_1-1}\frac{1}{N}\sum_{i=1}^N(Y_{N,i}- \bar{Y}_N)^2+ \frac{N-N_1}{(N-1)(N_1-1)}\frac{N-1}{N}S_N \right| \\
\leq & \frac{N_1}{N_1-1}\left| \frac{1}{N_1}\sum_{i=1}^NW_{N,i}(Y_{N,i}- \bar{Y}_N)^2- \frac{1}{N}\sum_{i=1}^N(Y_{N,i}- \bar{Y}_N)^2 \right| + \frac{N-N_1}{(N-1)(N_1-1)}\frac{N-1}{N}S_N \\
\leq & \frac{N_1}{N_1-1}\left| \frac{1}{N_1}\sum_{i=1}^NW_{N,i}(Y_{N,i}- \bar{Y}_N)^2- \frac{1}{N}\sum_{i=1}^N(Y_{N,i}- \bar{Y}_N)^2 \right| + \frac{1}{N_1-1}L^{1/2}\as  0.
\end{split}\]
We now finally have $| \hat{S}_N-S_N| \leq | \hat{S}_N- \tilde{S}_N|+| \tilde{S}_N-S_N| \as 0$.
\end{proof}
\begin{lem}\label{l:maxto0}
Under Assumption \ref{AsuA} and for all sequences of $W$, the imputed potential outcomes in FRT-\ref{step::impute} satisfy
$\lmt{N}\max_{i,j}\{ Y_i^*(j)- \bar{Y}^*(j) \}^2/N=0$.
\end{lem}
\begin{proof}
Recall the $z_j$'s in FRT-\ref{step::impute} and define $\bar{z} = \sum_{i=1}^N z_{W_i} / N= \sum_{j=1}^J  N_j z_j/N $.
Because $( \bar{Y}(j))$ converges for all $j=\ot{J}$, and the $z_j$'s do not depend on $N$, we may pick $Y_\T{max}\in \bb{R}$ such that for all $N$, 
\[ \max_j| \bar{Y}(j)| \vee \max_j|z_j- \bar{z}| \leq Y_\T{max}. \]
Put $L_N= \max_{i,j}\{ Y_i(j)- \bar{Y}(j) \}^2$, which is $o(N)$ by Assumption \ref{AsuA}. Then  
\[ \max_{i,j}|Y_i(j)- \bar{Y}(j)|= \big[ \max_{i,j}\{ Y_i(j)- \bar{Y}(j) \}^2 \big]^{1/2}\leq L_N^{1/2}. \]
Next,
\[ \max_i|Y_i^\T{obs}| \leq \max_{i,j}|Y_i(j)| \leq \max_{i,j}|Y_i(j)- \bar{Y}(j)|+ \max_j| \bar{Y}(j)|  
\leq L_N^{1/2}+Y_\T{max}. \]
Recall that $\bar{Y}_\cdot^\T{obs}= \sum_{i=1}^NY_i^\T{obs}/N$, and we have the following bounds:
\[ | \bar{Y}_\cdot^\T{obs}| \leq  \max_i|Y_i^\T{obs}| \leq L_N^{1/2}+Y_\T{max}, \quad
\max_i|Y_i^\T{obs}- \bar{Y}_\cdot^\T{obs}| \leq  \max_i|Y_i^\T{obs}|+| \bar{Y}_\cdot^\T{obs}| \leq 2(L_N^{1/2}+Y_\T{max}). \]
Using the above bounds and the additional bound $(a+b)^2 \leq 2(a^2+b^2)$, we have
\[ \max_i(Y_i^\T{obs}- \bar{Y}_\cdot^\T{obs})^2= \big( \max_i|Y_i^\T{obs}- \bar{Y}_\cdot^\T{obs}| \big)^2 \leq 4(L_N^{1/2}+Y_\T{max})^2 \leq 8(L_N+Y_\T{max}^2). \]
In FRT-\ref{step::impute}, we have $Y_i^*(j) = Y_i^\T{obs} + z_j - z_{W_i}$ and therefore $\bar{Y}^*(j) = \bar{Y}_\cdot^\T{obs} + z_j - \bar{z}$. Finally, we have
\[ \begin{split}
\max_{i,j}\{ Y_i^*(j)- \bar{Y}^*(j) \}^2=& \max_i(Y_i^\T{obs}-z_{W_i}- \bar{Y}_\cdot^\T{obs}+ \bar{z})^2 \\
\leq & 2 \big\{ \max_i(Y_i^\T{obs}- \bar{Y}_\cdot^\T{obs})^2+ \max_i(z_{W_i}- \bar{z})^2 \big\} \\
\leq & 16(L_N+Y_\T{max}^2)+2Y_\T{max}^2,
\end{split}\]
which is $o(N)$ as desired.
\end{proof}
Now we visit the vector versions of Lemmas \ref{l:asconv} and \ref{l:maxto0}.
\begin{lem}\label{l:asconv-vec}
Let $( \{ Y_{N,i}:i=\ot{N}\} ) $ be a sequence of populations with means $\bar{Y}_N \in \bb{R}^d$ and covariances $S_N$. Suppose we take a simple random sample from each population of size $N_1 \geq d+1$ with sample mean $\hat{\bar{Y}}_N$ and covariance $\hat{S}_N$. Assume $\lmt{N}N_1/N=p_1>0$.
\begin{enumerate}[(i)]
\item If the sequence $( \| S_N \|_F)$ is bounded above by $S_\T{max}< \infty$, then $| \hat{\bar{Y}}_N- \bar{Y}_N| \as 0$. If we also have $\lmt{N}\bar{Y}_N= \bar{Y}_\infty$, then $\hat{\bar{Y}}_N \as \bar{Y}_\infty$. Assumption \ref{AsuA-vec} implies these results.
\item If there is $L< \infty$ such that $\sum_{i=1}^N|Y_{N,i}- \bar{Y}_N|^4/N \leq L$ for all $N$, then $\| \hat{S}_N-S_N \|_F \as 0$. If we also have $\lmt{N}S_N=S_\infty$, then $\hat{S}_N \as S_\infty$. Assumption \ref{AsuB-vec} implies these results.
\end{enumerate}
\end{lem}
\begin{proof}
(i) Note that each component of $Y_{Ni}$ meets Lemma \ref{l:asconv}, so $| \hat{\bar{Y}}_N- \bar{Y}_N| \as 0$ holds component by component.

(ii) Because each component of $Y_{Ni}$ meets Lemma \ref{l:asconv}, each entry on the main diagonal of $\hat{S}_N-S_N$ converges almost surely to 0.
It is thus enough to show convergence of the $(1,2)$th entry, for then identical logic will show convergence of an arbitrary off-diagonal entry. Let $Y_{1Ni}$ and $Y_{2Ni}$ be the first and second entries of $Y_{Ni}$.
 
We follow the steps of Lemma \ref{l:asconv} closely. First, $\| S_N \|_F$ is bounded above:
\[ \begin{split}
 \| S_N \|_F=& \frac{1}{N-1}\big\| \sum_{i=1}^N \op{(Y_{Ni}- \bar{Y}_N)}\big\|_F \\
\leq & \frac{1}{N-1}\sum_{i=1}^N|Y_{Ni}- \bar{Y}_N|^2 
 \leq \frac{N^{1/2}}{N-1}\left( \sum_{i=1}^N|Y_{Ni}- \bar{Y}_N|^4 \right)^{1/2}\leq \frac{NL^{1/2}}{N-1},
\end{split}\]
where the first inequality follows from the Triangle Inequality and $\| ab^\transp \|_F=|a| \cdot |b|$ for two vectors $a$ and $b$, and the second inequality by the Cauchy--Schwarz Inequality. By (i), $| \hat{\bar{Y}}_N- \bar{Y}_N| \as 0$.

Second, let $W_{N,i}$ be the indicator for $Y_i$ being in the simple random sample. Define as an intermediate quantity $\tilde{S}_{12N}= \sum_{i=1}^N  W_{N,i}(Y_{1Ni}- \bar{Y}_{1N})(Y_{2Ni}- \bar{Y}_{2N})/(N_1-1)$, which differs from $\hat{S}_{12N}$ by an almost surely zero quantity as $N \to \infty$:
\[ \begin{split}
\hat{S}_{12N}- \tilde{S}_{12N}=& \frac{1}{N_1-1}\sum_{i=1}^NW_{N,i}\{ (Y_{1Ni}- \hat{\bar{Y}}_{1N})(Y_{2Ni}- \hat{\bar{Y}}_{2N})- (Y_{1Ni}- \bar{Y}_{1N})(Y_{2Ni}- \bar{Y}_{2N}) \} \\
=& \frac{1}{N_1-1}\sum_{i=1}^NW_{N,i} \{ ( \bar{Y}_{1N}- \hat{\bar{Y}}_{1N})Y_{2Ni}+( \bar{Y}_{2N}- \hat{\bar{Y}}_{2N})Y_{1Ni}+ \hat{\bar{Y}}_{1N}\hat{\bar{Y}}_{2N}- \bar{Y}_{1N}\bar{Y}_{2N}\} \\
=& \frac{N_1}{N_1-1}\{ ( \bar{Y}_{1N}- \hat{\bar{Y}}_{1N}) \hat{\bar{Y}}_{2N}+( \bar{Y}_{2N}- \hat{\bar{Y}}_{2N}) \hat{\bar{Y}}_{1N}+ \hat{\bar{Y}}_{1N}\hat{\bar{Y}}_{2N}- \bar{Y}_{1N}\bar{Y}_{2N}\} \\
=& \frac{-N_1}{N_1-1}( \bar{Y}_{1N}- \hat{\bar{Y}}_{1N})( \bar{Y}_{2N}- \hat{\bar{Y}}_{2N}) \as 0.
\end{split}\] 
Third, we note that the variance of $\{ (Y_{1Ni}- \bar{Y}_{1N})(Y_{2Ni}- \bar{Y}_{2N}) \}_{i=1}^N$ is bounded above for all $N$:
\[ \begin{split}
\Var \left[ \{ (Y_{1Ni}- \bar{Y}_{1N})(Y_{2Ni}- \bar{Y}_{2N}) \}_{i=1}^N \right]
\leq & \frac{1}{N-1}\sum_{i=1}^N(Y_{1Ni}- \bar{Y}_{1N})^2(Y_{2Ni}- \bar{Y}_{2N})^2 \\
\leq & \frac{1}{N-1}\left\{ \sum_{i=1}^N(Y_{1Ni}- \bar{Y}_{1N})^4 \sum_{i=1}^N(Y_{2Ni}- \bar{Y}_{2N})^4 \right\}^{1/2}\\
\leq & \frac{1}{N-1}\left\{ \sum_{i=1}^N(Y_{1Ni}- \bar{Y}_{1N})^4 \vee \sum_{i=1}^N(Y_{2Ni}- \bar{Y}_{2N})^4 \right\} \\
\leq & \frac{NL}{N-1}.
\end{split}\]
So by (i), $\big| \sum_{i=1}^NW_{N,i}(Y_{1Ni}- \bar{Y}_{1N})(Y_{2Ni}- \bar{Y}_{2N})/N_1- \sum_{i=1}^N(Y_{1Ni}- \bar{Y}_{1N})(Y_{2Ni}- \bar{Y}_{2N})/N \big| \as 0$. In addition, $S_{12N}\leq \| S_N \|_F$ is bounded from above. These imply that 
\[ \begin{split}
| \tilde{S}_{12N}-S_{12N}|=& \left| \frac{N_1}{N_1-1}\left\{ \frac{1}{N_1}\sum_{i=1}^NW_{N,i}(Y_{1Ni}- \bar{Y}_{1N})(Y_{2Ni}- \bar{Y}_{2N})- \frac{1}{N}\sum_{i=1}^N(Y_{1Ni}- \bar{Y}_{1N})(Y_{2Ni}- \bar{Y}_{2N}) \right\} \right. \\
& \left. + \left( \frac{N_1}{(N_1-1)N}- \frac{1}{N-1}\right) \sum_{i=1}^N(Y_{1Ni}- \bar{Y}_{1N})(Y_{2Ni}- \bar{Y}_{2N}) \right| \\
\leq & \frac{N_1}{N_1-1} \left| \frac{1}{N_1}\sum_{i=1}^NW_{N,i}(Y_{1Ni}- \bar{Y}_{1N})(Y_{2Ni}- \bar{Y}_{2N})- \frac{1}{N}\sum_{i=1}^N(Y_{1Ni}- \bar{Y}_{1N})(Y_{2Ni}- \bar{Y}_{2N}) \right| \\
&+ \frac{N-N_1}{(N-1)(N_1-1)}\frac{N-1}{N}S_{12N}  \as 0.
\end{split}\]
We now finally have $|\hat{S}_{12N}-S_{12N}| \leq |\hat{S}_{12N}- \tilde{S}_{12N}|+| \tilde{S}_{12N}-S_{12N}| \as 0$.
\end{proof}
\begin{lem}\label{l:maxto0-vec}
Under Assumption \ref{AsuA-vec} and for all sequences of $W$, the imputed potential outcomes satisfy $\lmt{N}\max_{i,j}|Y_i^*(j)- \bar{Y}^*(j)|^2/N=0$.
\end{lem}
\begin{proof}
From \eqref{e:impute-vec}, we obtain $\{ Y_i^*(j)_1:i=\ot{N},j=\ot{J}\}$ from $\{ W_i,(Y_i^\T{obs})_1:i=\ot{N}\}$ in the same way as FRT-\ref{step::impute}. So by Lemma \ref{l:maxto0}, we have $\lmt{N}\max_{i,j}\{ Y_i^*(j)_1- \bar{Y}^*(j)_1 \}^2/N=0$. Doing the same for the other $d-1$ entries gives the desired result.
\end{proof}

\section{Proofs of the Main Theorems}\label{s:pf}
We make some preliminary observations and extend the notation to handle the randomization distributions as required by Theorems \ref{t:CYx}, \ref{t:Brunner}, and \ref{t:LSF}. Throughout, we make heavy use of the mean of the observed values:
\[ \bar{Y}_\cdot^\T{obs}= \frac{1}{N}\sum_{i=1}^NY_i^\T{obs}= \sum_{j=1}^J \frac{N_j}{N}\hat{\bar{Y}}(j) \]
Recall the imputed potential outcomes FRT-\ref{step::impute} are $Y_i^*(j)=Y_i^\T{obs}+z_j-z_{W_i}$. They agree with the data in the sense $Y_i^*(W_i)=Y_i^\T{obs}$ for all $i=\ot{N}$.
They are also strictly additive, as $Y_i^*(j)-Y_i^*(k)= ( Y_i^\T{obs} + z_j-z_{W_i} ) -(Y_i^\T{obs} + z_k-z_{W_i})=z_j-z_k$ does not depend on the unit $i$. The imputed potential outcomes have means $\bar{Y}^*=( \bar{Y}^*(1), \ldots, \bar{Y}^*(J))^\transp$ and  covariance  $s^*1_J1_J^\transp$, due to strict additivity. Recalling that $\bar{z}= \sum_{j=1}^J N_jz_j/N$, we have
\begin{align} 
\bar{Y}^*(j)=& \frac{1}{N}\sum_{i=1}^N(Y_i^\T{obs}+z_j-z_{W_i})= \sum_{k=1}^J \frac{N_k}{N}\hat{\bar{Y}}(k)+z_j- \bar{z},
\label{e:imputemv}\\
s^*=& S^*(1,1)= \frac{1}{N-1}\sum_{i=1}^N \{ Y_i^*(1)- \bar{Y}^*(1) \}^2 \nonumber \\
=& \frac{1}{N-1}\sum_{j=1}^J \sum_{i=1}^NW_i(j) \{ Y_i^*(j)- \bar{Y}^*(j) \}^2 \label{e:constant} \\
=& \sum_{j=1}^J \frac{N_j-1}{N-1}\hat{S}(j,j)+ \sum_{j=1}^J \frac{N_j}{N-1}\{ \hat{\bar{Y}}(j)- \bar{Y}^*(j) \}^2, \label{e:biasvariance}
\end{align}
where \eqref{e:constant} follows from the facts that $Y_i^*(j)- \bar{Y}^*(j)$ does not depend on $j$ due to strict additivity and $\sum_{j=1}^JW_i(j)=1$, and \eqref{e:biasvariance} follows from the bias-variance decomposition (add and subtract $\hat{\bar{Y}}(j)$) and noting $Y_i^*(j)=Y_i^\T{obs}$ when $W_i=j$.

For asymptotic purposes, note that $C,x, \tilde{C}, \tilde{x}$ are fixed with respect to $N$, hence $z$ is as well. They may be regarded as constants as we take $N \to \infty$.

The analogs of $\hat{D}$ and $V$, for imputed potential outcomes are, respectively
\begin{equation}\label{e:PermCov}
\hat{D}_\pi=N \cdot \diag\{ \hat{S}_\pi(1,1)/N_1, \ldots, \hat{S}_\pi(J,J)/N_J \}, \quad V^*=s^*(P^{-1}-1_J1_J^\transp ).
\end{equation}
Compare these to \eqref{e:Dob} and \eqref{e:ObsCov}. We also have, conditional on $W$, that $\hat{D}_\pi-s^*P^{-1}\inP 0$. In general, consistent with previous patterns, analogs of population quantities have superscript ``$*$'', while those of observed quantities have subscript ``$\pi$''.
\begin{proof}[Proof of Theorems \ref{t:CYx}, \ref{t:Brunner}, and \ref{t:LSF}]
We prove the sampling, followed by the randomization distribution claims.
\paragraph{Sampling distributions of $X^2$, $F$, and $B$.}
Let Assumption \ref{AsuA} and $H_{0 \T{N}}(C,x)$ hold. We have $N^{1/2}(C \hat{\bar{Y}}-x) \ind \cl{N}(0_m,CVC^\transp )$, $C \hat{D}C^\transp \inP CDC^\transp \succ 0$ and $CDC^\transp \succeq CVC^\transp$ by Proposition \ref{t:Ding5} and \eqref{e:Dob}. Hence, by Lemma \ref{l:DD18-LA}
\[ X^2= N^{1/2}(C \hat{\bar{Y}}-x)^\transp (C \hat{D}C^\transp )^{-1}N^{1/2}(C \hat{\bar{Y}}-x) \ind \sum_{j=1}^ma_j \xi_j^2, \quad \T{ with }a_j \in [0,1] \quad (j=\ot{m}). \]
We deal with $B,F$ similarly. Assume $x=0_m$. By \eqref{e:Dob} and the Continuous Mapping Theorem, $\tr(M \hat{D})CC^\transp \inP \tr(MD)CC^\transp$. By Lemma \ref{l:DD18-LA},
\[ \begin{split}
B=& N^{1/2}(C \hat{\bar{Y}})^\transp ( \tr(M \hat{D})CC^\transp )^{-1}N^{1/2}C \hat{\bar{Y}}\ind \sum_{j=1}^m \lambda_j \big( CVC^\transp ( \tr(MD)CC^\transp )^{-1}\big) \xi_j^2 \\
\eqd & \sum_{j=1}^m \frac{1}{\tr(MD)}\lambda_j(VC^\transp (CC^\transp )^{-1}C) \xi_j^2 \eqd \frac{\sum_{j=1}^m \lambda_j(MV) \xi_j^2}{\tr(MD)}.
\end{split}\]
Recall $\cl{X}$ and $\hat{\sigma}^2$ in \eqref{e:Fstat}. Then $\hat{\sigma}^2 \inP \sum_{j=1}^Jp_jS(j,j)= \bar{S}$ by Proposition \ref{p:inP}, $(N_j-1)/(N-J) \to p_j$, and
\[ ( \cl{X}^\transp \cl{X}/N)^{-1}= \diag(N_1/N, \ldots,N_J/N)^{-1}\inP P^{-1}. \]
Therefore, by Lemma \ref{l:DD18-LA},
\[ mF= N^{1/2}(C \hat{\bar{Y}}) \{ \hat{\sigma}^2C( \cl{X}^\transp \cl{X})^{-1}C^\transp \}^{-1}N^{1/2}C \hat{\bar{Y}}\ind \sum_{j=1}^m \lambda_j(CVC^\transp ( \bar{S}CP^{-1}C^\transp )^{-1}) \xi_j^2. \]
\paragraph{Randomization distributions.}
We first show, for almost all realizations of the sequence of treatment assignments $W$, that Assumption \ref{AsuA} holds for $\{ U_i^*(j):i=\ot{N},j=\ot{J}\}$ where $U_i^*(j)= \{ Y_i^*(j)- \bar{Y}^*(j) \} /(s^*)^{1/2}$ are the standardized imputed potential outcomes.
Clearly they always have mean 0 and variance 1, so it is enough to verify that, almost surely
\begin{equation}\label{e:mainpf1}
\lmt{N}\max_{i,j}\frac{1}{N}\{ U_i^*(j)- \bar{U}^*(j) \}^2= \lmt{N}\max_{i,j}\frac{\{ Y_i^*(j)- \bar{Y}^*(j) \}^2}{Ns^*}=0.
\end{equation}
Starting with \eqref{e:biasvariance}, we have
$$
s^*= \sum_{j=1}^J \frac{N_j-1}{N-1}\hat{S}(j,j)+ \sum_{j=1}^J \frac{N_j}{N-1}\{ \hat{\bar{Y}}(j)- \bar{Y}^*(j) \}^2  
\geq  \frac{N_1-1}{N-1}\hat{S}(1,1) \as p_1S(1,1),
$$
where the last step is by Lemma \ref{l:asconv}. This shows the sequence $(s^*)_{N \geq 2J}$ is bounded away from 0, as $p_1>0$ and $S(1,1)>0$.
Now we also have $\lmt{N}N^{-1}\max_{i,j}\{ Y_i^*(j)- \bar{Y}^*(j) \}^2=0$, no matter what the realization of the sequence $\{ W \}_{N=1}^\infty$ is, by Lemma \ref{l:maxto0}. These two facts together show \eqref{e:mainpf1}.

Because $\hat{S}(1,1) \as S(1,1)$ by Lemma \ref{l:asconv}, we for the rest of the proof fix a sequence of $(W)$ along which $\hat{S}(1,1) \to S(1,1)$. The only remaining randomness then comes from $\pi \sim \Unif( \Pi_N)$.
Note for $i=\ot{N}$ that $CU_i^*=C(Y_i^*- \bar{Y}^*)/ (s^*)^{1/2}=0_m$ because $CY_i^*=x$ from the fact that the imputed potential outcomes satisfy \eqref{e:HFx}.
In particular, the standardized imputed potential outcomes satisfy $H_{0 \T{N}}(C,0_m)$, i.e., $C \bar{U}^*=0_m$. Hence, by Proposition \ref{t:Ding5}, we have
\[ \begin{split}
(N/s^*)^{1/2}(C \hat{\bar{Y}}_\pi -x)=& N^{1/2}C( \hat{\bar{Y}}_\pi- \bar{Y}^*)/(s^*)^{1/2}=N^{1/2}C \hat{\bar{U}}_\pi \\  
\ind & \cl{N}\big( 0_m,C(P^{-1}-1_J1_J^\transp )C^\transp \big) \eqd \cl{N}(0_m,CP^{-1}C^\transp )
\end{split}\]
because the standardized imputed potential outcomes have covariance structure $1_J1_J^\transp$ and $C1_J=0_m$. Next, for $j=\ot{J}$, we have
\[ \frac{\hat{S}_\pi(j,j)}{s^*}= \frac{1}{N_j-1}\sum_{i=1}^NW_{\pi(i)}(j) \frac{\{ Y_i^*(j)- \bar{Y}^*(j) \}^2}{s^*}= \frac{1}{N_j-1}\sum_{i=1}^NW_{\pi(i)}(j)U_i^*(j)^2 \inP 1 
 \]
by Proposition \ref{p:inP} and because the standardized imputed potential outcomes have variances 1. It follows by \eqref{e:PermCov} that
\[ \hat{D}_\pi/s^* \inP P^{-1}, \quad \hat{\sigma}_\pi^2/s^*= \sum_{j=1}^J \frac{N_j-1}{(N-J)s^*}\hat{S}_\pi(j,j) \inP 1,\quad \tr(M\hat{D}_\pi)/s^* \inP \tr(MP^{-1}). \]
We thus finally have by Lemma \ref{l:DD18-LA}
\[ X_\pi^2=(N/s^*)^{1/2}(C \hat{\bar{Y}}_\pi -x)^\transp (C \hat{D}_\pi C^\transp /s^*)^{-1}(N/s^*)^{1/2}(C \hat{\bar{Y}}_\pi -x) \ind \sum_{j=1}^m \lambda_j \big( CP^{-1}C^\transp (CP^{-1}C^\transp)^{-1}\big) \xi_j^2 \eqd \chi_m^2, \]
and with $x=0_m$ for the $B$ and $F$ statistics:
\[ \begin{split} 
B_\pi=& (N/s^*)^{1/2}(C \hat{\bar{Y}}_\pi)^\transp \{ \tr(M \hat{D}_\pi)CC^\transp /s^* \}^{-1}(N/s^*)^{1/2}C \hat{\bar{Y}}_\pi \\
\ind & \sum_{j=1}^m \lambda_j \big( CP^{-1}C^\transp ( \tr(MP^{-1})CC^\transp )^{-1}\big) \xi_j^2 \eqd \sum_{j=1}^m \lambda_j(MP^{-1}) \xi_j^2/ \tr(MP^{-1}), \\
mF_\pi=& (N/s^*)^{1/2}(C \hat{\bar{Y}}_\pi )^\transp \left\{ \frac{\hat{\sigma}_\pi^2}{s^*}C( \cl{X}^\transp \cl{X}/N)^{-1}C^\transp \right\}^{-1}(N/s^*)^{1/2}C \hat{\bar{Y}}_\pi \\
\ind & \sum_{j=1}^m \lambda_j \big( CP^{-1}C^\transp (CP^{-1}C^\transp)^{-1} \big) \xi_j^2 \eqd \chi_m^2. \qedhere
\end{split}\]
\end{proof}
Extending Theorem \ref{t:CYx} to the case of stratified experiments or vector potential outcomes is straightforward. We also supply their proofs for completeness.
\begin{proof}[Proof of Theorem \ref{t:CYx-bl}]
We prove the sampling, followed by the randomization distribution claims.
\paragraph{Sampling distribution of $X^2$.}
For $h=\ot{H}$, we have that $\bb{E} ( \hat{\bar{Y}}_{[h]} ) = \bar{Y}_{[h]}$, and that Assumption \ref{AsuA} holds in each stratum $h$. By Proposition \ref{t:Ding5},
\[ N_{[h]}^{1/2}C( \hat{\bar{Y}}_{[h]}- \bar{Y}_{[h]}) \ind \cl{N}(0_m,CV_{[h]}C^\transp ), \T{ where }
V_{[h]}=  \plim_{N \to \infty}\hat{D}_{[h]}-S_{[h]}. \]
Under $H_{0 \T{N}}(C,x)$, we have $x=C \bar{Y}= \sum_{h=1}^HN_{[h]}C \bar{Y}_{[h]}/N$. Because $(\hat{\bar{Y}}_{[1]}, \ldots, \hat{\bar{Y}}_{[H]})$ are mutually independent in a SRE, we have
\[ \begin{split}
N^{1/2}(C \breve{\bar{Y}}-x)=& \sum_{h=1}^H \left(\frac{N_{[h]}}{N}\right)^{1/2}N_{[h]}^{1/2}C( \hat{\bar{Y}}_{[h]}- \bar{Y}_{[h]}) \\
\ind & \sum_{h=1}^H \omega_{[h]}^{1/2}\cl{N}(0_m,CV_{[h]}C^\transp ) \eqd \cl{N}\left( 0_m, \sum_{h=1}^H \omega_{[h]}CV_{[h]}C^\transp \right) .
\end{split}\]
Next, note that $\plim_{N \to \infty}\hat{D}_{[h]}\succeq V_{[h]}$ implies $\plim_{N \to \infty}\sum_{h=1}^HN_{[h]}C\hat{D}_{[h]}C^\transp /N \succeq  \sum_{h=1}^H \omega_{[h]}CV_{[h]}C^\transp$, so by Lemma \ref{l:DD18-LA}, we have
\[ X^2=N^{1/2}(C \breve{\bar{Y}}-x)^\transp \left( C \sum_{h=1}^H \frac{N_{[h]}}{N}\hat{D}_{[h]}C^\transp \right)^{-1}N^{1/2}(C \breve{\bar{Y}}-x) \ind \sum_{j=1}^ma_j \xi_j^2. \]
\paragraph{Randomization distribution of $X^2$.}
We first show Assumption \ref{AsuA} holds almost surely within each stratum for the imputed potential outcomes $Y_i^*(j)$. Because the original potential outcomes satisfy Assumption \ref{AsuA} in each stratum, Lemma \ref{l:maxto0} gives $\lmt{N}\max_j \max_{i:X_i=h}\{ Y_i^*(j)- \bar{Y}_{[h]}^*(j) \}^2/N_{[h]}=0$.
Put $\bar{z}_{[h]}= \sum_{j=1}^JN_{[h]j}z_{[h],j}/N_{[h]}$. In stratum $h$, the mean vector is $\bar{Y}_{[h]}^*$ and the covariance structure is $s_{[h]}^*1_J1_J^\transp$, where
\[ \begin{split}
\bar{Y}_{[h]}^*(j)=& \sum_{k=1}^J \frac{N_{[h]k}}{N_{[h]}}\hat{\bar{Y}}_{[h]}(k)+z_{[h],j}- \bar{z}_{[h]}\\
s_{[h]}^*=& \sum_{j=1}^J \frac{N_{[h]j}-1}{N_{[h]}-1}\hat{S}_{[h]}(j,j)+ \sum_{j=1}^J \frac{N_{[h]j}}{N_{[h]}-1}\{ \hat{\bar{Y}}_{[h]} (j)- \bar{Y}_{[h]}^*(j) \}^2,
\end{split}\]
by applying \eqref{e:imputemv} and \eqref{e:biasvariance} to stratum $h$. $\hat{\bar{Y}}_{[h]}(j)$ and $\hat{S}_{[h]}(j,j)$ converge almost surely because of Lemma \ref{l:asconv}, applicable because Assumption \ref{AsuB} holds within stratum $h$. Then $\bar{Y}_{[h]}^*(j)$ and $s_{[h]}^*$ converge almost surely because all quantities on the right-hand side do.
This shows Assumption \ref{AsuA} holds within each stratum almost surely.

For the rest of the proof, fix a sequence $(W)$ along which $(s_{[h]}^*)$ converges. Because each $CY_i^*=x_{[h]}$ whenever $X_i=h$, we have $C \bar{Y}_{[h]}^*=x_{[h]}$, and by Proposition \ref{t:Ding5},
\[ N_{[h]}^{1/2}C( \hat{\bar{Y}}_{[h],\pi}- \bar{Y}_{[h]}^*) \ind \cl{N}\big( 0_m,s_{[h]}^*C(P^{-1}-1_J1_J^\transp )C^\transp \big) \eqd \cl{N}(0_m,s_{[h]}^*CP^{-1}C^\transp ). \]
Since $x= \sum_{h=1}^HN_{[h]}x_{[h]}/N= \sum_{h=1}^HN_{[h]}C \bar{Y}_{[h]}^*/N$, it follows that
\[ \begin{split}
N^{1/2}(C \breve{\bar{Y}}_\pi-x)=& \sum_{h=1}^H \left( \frac{N_{[h]}}{N}\right)^{1/2}N_{[h]}^{1/2}C( \hat{\bar{Y}}_{[h],\pi}- \bar{Y}_{[h]}^*) \\
\ind &  \sum_{h=1}^H \omega_{[h]}^{1/2}\cl{N}\left(0_m,s_{[h]}^*CP^{-1}C^\transp \right) \eqd \cl{N}\left( 0_m, \sum_{h=1}^H \omega_{[h]}s_{[h]}^*CP^{-1}C^\transp \right)
\end{split}\]
because, conditioning on $W$, the $(\hat{\bar{Y}}_{[1],\pi}, \ldots, \hat{\bar{Y}}_{[H],\pi})$ are mutually independent. 
Next, from Proposition \ref{p:inP}, we have $\hat{D}_{[h],\pi}\inP s_{[h]}^*P^{-1}$, so $C \sum_{h=1}^HN_{[h]}\hat{D}_{[h],\pi}C^\transp /N \inP \sum_{h=1}^H \omega_{[h]}s_{[h]}^*CP^{-1}C^\transp$, and we finally have from Lemma \ref{l:DD18-LA}
\[ X_\pi^2=N^{1/2}(C \breve{\bar{Y}}_\pi-x)^\transp \left( C \sum_{h=1}^H \frac{N_{[h]}}{N}\hat{D}_{[h],\pi}C^\transp \right)^{-1}N^{1/2}(C \breve{\bar{Y}}_\pi-x) \ind \chi_m^2. \qedhere \]
\end{proof}
\begin{proof}[Proof of Theorem \ref{t:MPO}]
We prove the sampling, followed by the randomization distribution claims.
\paragraph{Sampling distribution of $X^2$.}
Under Assumption \ref{AsuA-vec} and $H_{0 \T{N}}(C,x)$, we use \cite{DingCLT} to prove the following results in parallel with Propositions \ref{p:inP} and \ref{t:Ding5}. First, $\hat{\bar{Y}}\inP \bar{Y}$ and $\hat{S}(j,j) \inP S(j,j)$ for $j=\ot{J}$. Second, $N^{1/2}(C \hat{\bar{Y}}-x) \ind \cl{N}(0_m,CVC^\transp )$, where we have the vector potential outcomes analog of \eqref{e:ObsCov}:
\begin{equation}\label{e:ObsCov-vec}
V= \lmt{N}N \cdot \Cov( \hat{\bar{Y}})= \lmt{N}\begin{pmatrix}
\frac{N-N_1}{N_1}S(1,1) & -S(1,2) & \cdots & -S(1,J) \\
-S(2,1) & \frac{N-N_2}{N_2}S(2,2) & \cdots & -S(2,J) \\
\vdots & \vdots & \ddots & \vdots \\
-S(J,1) & -S(J,2) & \cdots& \frac{N-N_J}{N_J}S(J,J) \end{pmatrix}.
\end{equation}
Because $C\hat{D} C^\transp \inP C(V+S)C^\transp \succeq CVC^\transp$, it follows from Lemma \ref{l:DD18-LA} that $X^2=N(C \hat{\bar{Y}}-x)^\transp (C\hat{D} C^\transp )^{-1}(C \hat{\bar{Y}}-x) \ind \sum_{j=1}^ma_j \chi_j^2$.
\paragraph{Randomization distribution of $X^2$.}
We first show Assumption \ref{AsuA-vec} holds almost surely for the imputed potential outcomes $Y_i^*(j)$. Because the original potential outcomes satisfy Assumption \ref{AsuA-vec}, Lemma \ref{l:maxto0-vec} gives $\lmt{N}\max_{i,j}| Y_i^*(j)- \bar{Y}^*(j)|^2/N=0$.
Their means satisfy
\[ 
\bar{Y}^*(j)_1= \frac{1}{N}\sum_{i=1}^N( Y_{i,1}^\T{obs} + z_{1j} -z_{1,W_i})
= \frac{1}{N}\sum_{k=1}^JN_j \hat{\bar{Y}}(k)_1 +  z_{1j}- \bar{z}_1 ,
\]
where $\bar{z}_1 = \sum_{j=1}^J N_j  z_{1j}/N.$ 
Hence, the $\bar{Y}^*(j)_1$ converge almost surely because $\hat{\bar{Y}}(j) \as \bar{Y}(j)$ by Lemma \ref{l:asconv-vec}. By the same reasoning, the other entries of $\bar{Y}^*(j)$ also converge almost surely.
The covariance structure of the imputed potential outcomes is $(1_J1_J^\transp ) \otimes S^*(1,1)$, where following the same steps to derive \eqref{e:biasvariance}, we get
\[ \begin{split}
S^*(1,1)=& \frac{1}{N-1}\sum_{i=1}^N \op{\{ Y_i^*(1)- \bar{Y}^*(1) \}}\\
=& \frac{1}{N-1}\sum_{j=1}^J \sum_{i=1}^NW_i(j) \op{\{ Y_i^*(j)- \bar{Y}^*(j) \}}\\
=& \sum_{j=1}^J \frac{N_j-1}{N-1}\hat{S}(j,j)+ \sum_{j=1}^J \frac{N_j}{N-1}\op{\{ \hat{\bar{Y}}(j)- \bar{Y}^*(j) \}}.
\end{split}\]
This converges almost surely because all quantities in the last line do. For instance, $\hat{S}(j,j)$ converge almost surely because of Lemma \ref{l:asconv-vec}, applicable because of Assumption \ref{AsuB-vec}. This shows Assumption \ref{AsuA-vec} holds almost surely.

For the rest of the proof, fix a sequence $ ( W) $ along which Assumption \ref{AsuA-vec} is met. The limit of $S^*(1,1)$ must be invertible because the above calculation shows $S^*(1,1) \succeq (N_1-1)S(1,1)/(N-1) \succ 0$. Because each $CY_i^*=x$, the vector potential outcomes analog of Proposition \ref{t:Ding5} gives us
\[ \begin{split}
N^{1/2}(C \hat{\bar{Y}}_\pi-x)=N^{1/2}C( \hat{\bar{Y}}_\pi- \bar{Y}^*) \ind & \cl{N}\big( 0_m,C \{ (P^{-1}-1_J1_J^\transp ) \otimes S^*(1,1) \} C^\transp \big) \\
\eqd & \cl{N}\big( 0_m,C \{ P^{-1}\otimes S^*(1,1) \} C^\transp \big) .
\end{split}\]
The cancellation in the last line occurred, for instance because the $(1,2)$-block of $C \{ (1_J1_J^\transp ) \otimes S^*(1,1) \} C^\transp$ is $(C_1 \otimes e_1^\transp ) \{ (1_J1_J^\transp ) \otimes S^*(1,1) \} (C_2 \otimes e_2^\transp )^\transp =(C_11_J1_J^\transp C_2^\transp ) \otimes \{ e_1^\transp S^*(1,1)e_2 \}$, which vanishes because $C_1,C_2$ are themselves contrast matrices. Next,
\[ \hat{D}_\pi \inP \diag\left\{ \frac{S^*(1,1)}{p_1}, \ldots, \frac{S^*(1,1)}{p_J}\right\} =P^{-1}\otimes S^*(1,1), \]
so $C \hat{D}_\pi C^\transp \inP C \{ P^{-1}\otimes S^*(1,1) \} C^\transp$, and we finally have from Lemma \ref{l:DD18-LA} that $X^2_\pi=N(C \hat{\bar{Y}}_\pi-x)^\transp (C\hat{D}_\pi C^\transp )^{-1}(C \hat{\bar{Y}}_\pi-x) \ind \chi_m^2$.
\end{proof}

\section{Proofs of other results}\label{a::proofothers}
\begin{proof}[Proof of Proposition \ref{p:BimpA}]
The conclusion follows from
\[ \begin{split}
\max_{i,j}\frac{1}{N}\{ Y_i(j)- \bar{Y}(j) \}^2=& \frac{1}{N}\left[ \max_{i,j}\{ Y_i(j)- \bar{Y}(j) \}^4 \right]^{1/2}\\
\leq & \frac{1}{N}\left[ \max_j \sum_{i=1}^N \{ Y_i(j)-\bar{Y}(j) \}^4 \right]^{1/2}\leq (L/N)^{1/2}
\end{split}\]
which converges to $0$ as $N \to \infty$.
\end{proof}

\begin{proof}[Proof of Proposition \ref{p:inP}]
It follows from Theorem 1 and Proposition 1 of \citet{DingCLT}. 
\end{proof}

\begin{proof}[Proof of Proposition \ref{t:Ding5}]
It follows from Theorem 5 of of \citet{DingCLT}.
\end{proof}

\begin{proof}[Proof of Proposition \ref{p:suitable}] 
Assume $H_{0 \T{N}}(C,x)$ throughout. Let $U \sim \Unif(0,1)$.  Define
\[ 
F(x) = \bb{P}(T \leq x), \quad G(x) =  \bb{P}(T<x), \quad F_W(x) =  \bb{P}(T_\pi \leq x|W), \quad G_W(x) = \bb{P}(T_\pi <x|W). 
\]
Fix $\alpha \in (0,1)$. Note $G_W(T)=(N!)^{-1}\sum_{\pi \in \Pi_N}1(T_\pi<T)$, so
\[ \bb{P}\left\{ \frac{1}{N!}\sum_{\pi \in \Pi_N}1(T_\pi \geq T) \leq \alpha \right\} = \bb{P}\{ 1-G_W(T) \leq \alpha \} \leq \bb{P}\{ G(T) \geq 1- \alpha \} \leq \bb{P}(U \geq 1- \alpha)= \alpha \]
where we have used $T \leq_\T{st}T_\pi |W$ if and only if $G_W \leq G$ on $\bb{R}$ and $G(T) \leq_\T{st}U$.

\end{proof}

\begin{proof}[Proof of Corollary \ref{c:Brunner}]
First, if $S(1,1)= \cdots=S(J,J)$, then $D=S(1,1)P^{-1}$ from \eqref{e:Dob}. Recall from $V \preceq D$ that each $\lambda_j(MV) \leq \lambda_j(MD)$.
Therefore, under $H_{0 \T{N}}(C,x)$, Theorem \ref{t:Brunner} implies that
\[ B \ind \frac{\sum_{j=1}^m \lambda_j(MV) \xi_j^2}{\tr(MD)}\leq_\T{st}\frac{\sum_{j=1}^m \lambda_j(MD) \xi_j^2}{\tr(MD)}= \frac{\sum_{j=1}^m S(1,1) \lambda_j(MP^{-1}) \xi_j^2}{S(1,1)\tr(MP^{-1})}= \frac{\sum_{j=1}^m \lambda_j(MP^{-1}) \xi_j^2}{\tr(MP^{-1})}\eqd B_\pi|W. \]
So the criterion of Proposition \ref{p:suitable} is met.

Second, if $C$ is a row vector, then $M=C^\transp C/CC^\transp$. Therefore
\[ B= \frac{\hat{\bar{Y}}^\transp C^\transp C \hat{\bar{Y}}/CC^\transp}{\tr(C^\transp C \hat{D})/CC^\transp}= \frac{(C \hat{\bar{Y}})^\transp C \hat{\bar{Y}}}{C \hat{D}C^\transp}=(C \hat{\bar{Y}})^\transp (C \hat{D}C^\transp )^{-1}C \hat{\bar{Y}}=X^2. \qedhere \]
\end{proof}
\begin{proof}[Proof of Proposition \ref{p:BeqF}]
Under a balanced design we have $N_1= \ldots=N_J=N/J$, $X^\transp X=N_1I_J$ and $\hat{\sigma}^2= \sum_{j=1}^J \hat{S}(j,j)/J$. Thus, $F=N_1 \hat{\bar{Y}}^\transp M \hat{\bar{Y}}/(m \hat{\sigma}^2)$.
If $M$ has the same values on its main diagonal, then each value is in fact $m/J$ because the trace and rank of a projection matrix are the same. This implies
\[ \frac{N}{\tr(M \hat{D})}= N/ \left\{ \sum_{j=1}^J \frac{N}{N_j}\hat{S}(j,j) \frac{m}{J}\right\} = \frac{N}{m \sum_{j=1}^J \hat{S}(j,j)}= \frac{N_1}{m \hat{\sigma}^2}\Longrightarrow
B= \frac{N( \hat{\bar{Y}})^\transp M \hat{\bar{Y}}}{\tr(M \hat{D})}= \frac{N_1 \hat{\bar{Y}}^\transp M \hat{\bar{Y}}}{m \hat{\sigma}^2}=F. \qedhere \]
\end{proof}
\begin{proof}[Proof of Corollary \ref{c:LSF}]
If $S(1,1)= \cdots=S(J,J)$, then $\bar{S}= \sum_{j=1}^Jp_jS(j,j)=S(1,1)$ and $D= \bar{S}\cdot P^{-1}$. Therefore, $0 \leq \lambda_j \big( CVC^\transp( \bar{S}CP^{-1}C^\transp)^{-1}\big) = \lambda_j \big( CVC^\transp( CDC^\transp)^{-1}\big) \leq 1$ because $V \preceq D$.
By Theorem \ref{t:LSF}, under $H_{0 \T{N}}(C,0_m)$, we have
\[ m \cdot F \ind \sum_{j=1}^m\lambda_j \big( CVC^\transp( \bar{S}CP^{-1}C^\transp)^{-1}\big) \xi_j^2 \leq_\T{st}\chi_m^2, \qquad
m \cdot F_\pi|W \ind \chi_m^2. \qedhere \]
\end{proof}
\begin{proof}[Proof of Proposition \ref{p:eqX2}]
The conclusions follow from simple linear algebra facts. They seem to be known, but we give a proof for completeness. 

We first equate the $X^2$. As stated, in the ANOVA setting, $C=(1_{J-1}, \; -I_{J-1})$ and $x= 0_{J-1}$. Put $Q_j=N_j/ \hat{S}(j,j)$ and $Q= \sum_{j=1}^JQ_j$. Then by block matrix multiplication
\[ \frac{1}{N}C\hat{D}C^\transp =(1_{J-1}, \; -I_{J-1}) \diag(1/Q_1, \ldots,1/Q_J) \vect{1_{J-1}^\transp}{-I_{J-1}}= \frac{1}{Q_1}1_{J-1}1_{J-1}^\transp + \diag(1/Q_2, \ldots,1/Q_J). \]
Thus, using the Sherman--Morrison formula, we have
\[ \begin{split}
\left( \frac{1}{N}C\hat{D}C^\transp \right)^{-1} 
=& \diag(Q_2, \ldots, Q_J)- \left\{ \frac{1}{Q_1}\vecth{Q_2}{\vdots}{Q_J}(Q_2, \ldots,Q_J) \right\} 
\Big/ \left\{ 1+ \frac{1}{Q_1}\sum_{j=2}^JQ_j \right \} \\
=& \diag(Q_2, \ldots,Q_J)- \frac{1}{Q}\vecth{Q_2}{\vdots}{Q_J}(Q_2, \ldots,Q_J).
\end{split} \]
Finally, from \eqref{e:X2}, we have
\[ \begin{split}
X^2=& {\footnotesize \big( \hat{\bar{Y}}(1)- \hat{\bar{Y}}(2), \ldots, \hat{\bar{Y}}(1)- \hat{\bar{Y}}(J) \big) 
\left\{ \diag(Q_2, \ldots,Q_J)- \frac{1}{Q}\vecth{Q_2}{\vdots}{Q_J}(Q_2, \ldots, Q_J) \right\} \vecth{\hat{\bar{Y}}(1)- \hat{\bar{Y}}(2)}{\vdots}{\hat{\bar{Y}}(1)- \hat{\bar{Y}}(J)}}\\
=& \sum_{j=2}^JQ_j \{ \hat{\bar{Y}}(1)- \hat{\bar{Y}}(j) \}^2- \frac{1}{Q}\left[ \sum_{j=2}^JQ_j \{ \hat{\bar{Y}}(1)- \hat{\bar{Y}}(j) \} \right]^2.
\end{split}\]
Now we recognize the expression in \eqref{e:DD18-X2} as $Q$ times the variance of $\{ \hat{\bar{Y}}(1), \ldots, \hat{\bar{Y}}(J) \}$ under the probabilities $Q_1/Q, \ldots, Q_J/Q$.
But variance is unaffected by switching signs, and then adding the constant $\hat{\bar{Y}}(1)$ to all quantities, so \eqref{e:DD18-X2} is $Q$ times the variance of $\{ 0, \hat{\bar{Y}}(1)-\hat{\bar{Y}}(2), \ldots, \hat{\bar{Y}}(1)-\hat{\bar{Y}}(J) \}$ under the same probabilities, which is precisely what $X^2$ is above.

Next, we equate the $F$. Recall that $m=J-1$. It is thus enough to show
\[ (C \hat{\bar{Y}})^\transp \{ C( \cl{X}^\transp \cl{X})^{-1}C^\transp \}^{-1}C \hat{\bar{Y}}= \sum_{j=1}^JN_j \{ \hat{\bar{Y}}(j)- \bar{Y}_\cdot^\T{obs}\}^2. \]
This follows an identical argument to showing the $X^2$ coincide, with $N_j,N$ in place of $Q_j,Q$.
\end{proof}

\begin{proof}[Proof of Corollary \ref{c:tc}]
The expression \eqref{e:X2-tc} follows by matrix algebra. Because $C = (1, -1)$ is a row vector, Corollary \ref{c:Brunner} implies $B=X^2$, which is proper. \citet{DD18} have proved the rest of the corollary.
\end{proof}

\begin{proof}[Proof of Corollary \ref{c:CYx-ineq}]
Under Assumption \ref{AsuA} and $H_{0 \T{N}}(C,x)$ with a row vector $C$, we have $N^{1/2}(C \hat{\bar{Y}}-x) \ind \cl{N}(0,CVC^\transp )$ by Proposition \ref{t:Ding5}, $C\hat{D}C^\transp \inP CDC^\transp >0$ and $CDC^\transp \geq CVC^\transp$ by \eqref{e:Dob}. Hence,  
\[ t= \frac{N^{1/2}(x-C \hat{\bar{Y}})}{(C\hat{D}C^\transp )^{1/2}}\ind \cl{N}(0,a), \T{ where }a= \frac{CVC^\transp}{CDC^\transp}\in [0,1]. \]
To show the randomization distribution under Assumption \ref{AsuB}, we have $\hat{S}(1,1) \as S(1,1)$ by Lemma \ref{l:asconv}, so fix a sequence of $(W)$ along which $\hat{S}(1,1) \to S(1,1)$.
Then $(N/s^*)^{1/2}(C \hat{\bar{Y}}_\pi-x) \ind \cl{N}(0,CP^{-1}C^\transp )$ and $\hat{D}_\pi/s^* \inP P^{-1}$ (these are intermediate steps in the proof of Theorem \ref{t:CYx}), so
\[ t_\pi|W= \frac{N^{1/2}(x-C \hat{\bar{Y}}_\pi)}{(C\hat{D}_\pi C^\transp )^{1/2}}=(N/s^*)^{1/2}\frac{x-C \hat{\bar{Y}}_\pi}{(C\hat{D}_\pi C^\transp )^{1/2}}\ind \cl{N}(0,1). \]

To argue $t_+$ is proper for \eqref{e:HNx-ineq}, we let $x=x_0$. Then we want to test $\tilde{H}_{0 \T{N}}(C,x_0):C \bar{Y}\geq x_0$. The notation switch frees up $x$ as a dummy variable.
Let $p(x)$ be the $p$-value from testing $C \bar{Y}=x$ with $t_+=t_+(x)$. Then the $p$-value for $\tilde{H}_{0 \T{N}}(C,x_0)$ is $\sup_{x \geq x_0}p(x)$.
When $x \leq C \hat{\bar{Y}}$, we have $t_+=0$, so $p(x)=1$. If $C \hat{\bar{Y}}\geq x_0$, then $t_+(x_0)=0$, so $p(x_0)=1$ (see also the Hodges--Lehmann discussion), and $\sup_{x \geq x_0}p(x)=1=p(x_0)$.
The more interesting case is $C \hat{\bar{Y}}<x_0$. Then $t_+(x_0) \leq t_+(x)$ when $x \geq x_0$. The fact that $t_\pi(x)|W \ind \cl{N}(0,1)$ a.s. for all $x \in \bb{R}$ suggests asymptotically that $p(x_0) \geq p(x)$ when $x \geq x_0$, so $\sup_{x \geq x_0}p(x)=p(x_0)$.
Asymptotically speaking, we thus always have $\sup_{x \geq x_0}p(x)=p(x_0)$. This is why we can test $\tilde{H}_{0 \T{N}}(C,x)$ with $t_+$ as if we were testing $H_{0 \T{N}}(C,x)$.
\end{proof}

\begin{proof}[Proof of Proposition \ref{p:BimpA-vec}]
We omit it because it is similar to the proof of Proposition \ref{p:BimpA}. 
\end{proof}
\end{document}